\documentclass[10pt]{article}

\pdfoutput=1
\usepackage[english]{babel} 
\usepackage{amsfonts} 
\usepackage{amssymb} 
\usepackage{amsmath}
\usepackage{mathrsfs}
\usepackage{eucal}
\usepackage{graphicx}
\usepackage{makeidx}
\usepackage{color}
\usepackage{url}
\usepackage{appendix}
\usepackage[T1]{fontenc}
\usepackage[utf8]{inputenc}

\begin{document}

\newenvironment{proof}{{\noindent\bf Proof.}}{\hfill $\Box$ \medskip}

\newtheorem{teo}{Theorem}[section]
\newtheorem{defi}[teo]{Definition}
\newtheorem{pro}[teo]{Proposition}
\newtheorem{lemma}[teo]{Lemma}
\newtheorem{coro}[teo]{Corollary}
\newtheorem{rem}[teo]{Remark}
\newtheorem{cond}[teo]{Condition}
\newtheorem{ass}[teo]{Assumption}

\newcommand{\red}[1]{\textcolor{red}{#1}}

\def \non{{\nonumber}}
\def \noin{{\noindent}}
\def \hat{\widehat}
\def \tilde{\widetilde}
\def \bar{\overline}
\def \P{\mathbb{P}}
\def \E{\mathbb{E}}
\def \R{\mathbb{R}}
\def \Z{\mathbb{Z}}
\def \1{\mathbf 1}

\renewcommand {\theequation}{\arabic{section}.\arabic{equation}}

\title{\large {\bf Inflation, ECB and short-term interest rates:}\\{\bf A new model, with calibration to market data}}
 
\author{
Flavia Antonacci 
\footnote{Department of Economic Studies, University of Chieti-Pescara, Pescara, Italy, CORRESPONDING AUTHOR, flavia.antonacci@unich.it}
\and 
Cristina Costantini
\footnote{Department of Economic Studies and Unit\`a locale INdAM, , University of Chieti-Pescara, Pescara, Italy}
\and 
Fernanda D'Ippoliti
\footnote{RBC Capital Markets, London, UK.}
\and 
Marco Papi
\footnote{School of Engineering, UCBM, Rome, Italy}
 }

\date{August 13, 2020}
\maketitle
     
\begin{abstract}
We propose a new model for the joint evolution 
of the European inflation rate,
the European Central Bank official interest rate and the short-term interest
rate, in a stochastic, continuous time setting. 

We derive the valuation equation for a contingent claim and show that it has a unique solution. The contingent claim payoff may depend on all three economic factors of the model  and the discount factor is allowed to 
include inflation.

Taking as a benchmark the model  
of Ho, H.W., Huang, H.H. and Yildirim, Y., Affine model of inflation-indexed derivatives and inflation risk premium, 
{\em European Journal of Operational Research}, 2014, we show that our 
model performs better on market data from 2008 to 2015. 

Our model is not an affine model. Although in some special cases the solution of the valuation equation might admit a closed form, in general it has to be solved numerically. This can be done efficiently by the algorithm that we provide. Our model uses many fewer parameters than the benchmark model, which partly compensates the higher complexity of the numerical procedure and also 
suggests that our model describes the behaviour of the economic factors more closely. 

\end{abstract}


\vspace*{0.25cm}


\vspace*{0.25cm}


\vspace*{0.5cm}


\setcounter{equation}{0}
\section{Introduction}

The issuance of sovereign bonds linked to euro area
inflation began with the introduction of bonds indexed to the French
consumer price index (CPI) excluding tobacco (Obligations
Assimilables du Tr\'{e}sor index\'{e}es or OATis) in 1998. In 2003,
Greece, Italy and Germany decided to issue inflation-linked 
bonds too, and nowadays a variety of inflation-linked 
derivatives are quoted in European financial 
markets. Typical examples are {\em inflation indexed }
{\em swaps}, which allow to exchange 
the inflation rate for a fixed interest rate, or {\em inflation caps}, 
which pay out if the inflation exceeds a certain threshold over a given
period (for a detailed introduction to inflation
derivatives, see e.g. \cite{DeaconDerryMirfendereski2004}). 

The pricing of inflation-linked derivatives is related to both interest
rate and foreign exchange theory. In their seminal work 
of 2003, Jarrow and
Yildirim~\cite{JarrowYildirim2003} proposed an approach based
on foreign currency and interest rate derivatives valuation. On the
other hand, there is some empirical and theoretical evidence
that bond prices, inflation, interest rates, monetary policy and
output growth are related. In particular, both inflation 
and interest rates 
are clearly related to the activity of central banks. 

In the present work we propose a model for the joint 
evolution of the European inflation rate, the European Central Bank
(henceforth ECB) official interest rate and the short-term
interest rate, and use it to price European type 
derivatives whose payoff depends potentially on all three factors. 
To the best of our knowledge, ours is the 
first model that takes into account the interaction 
among all these three factors. 
With the 2007-2008 financial crisis it has become clear 
that there is another risk factor underlying bond prices, 
namely credit risk, but we leave the construction 
of a model that incorporates this factor for 
future work. 

Our model is a stochastic, continuous time one. 
More precisely, the ECB interest rate evolves as a pure jump process 
with jump intensity and distribution that depend both on its current value 
and on the current value of inflation. 
As inflation is measured at regular times, 
the inflation rate is modeled as a piecewise constant process that jumps 
at fixed times $t_i$. 
The new value at $t_i$ is given 
by a Gaussian random variable with expectation depending 
on the previous value of the inflation rate and on the current 
value of the ECB interest rate. 
Finally, the short-term interest rate follows a CIR type model 
with reversion towards an affine function of the ECB 
interest rate and diffusion coefficient depending on the spread 
between itself and the ECB interest rate. 

Many models proposed to price inflation indexed
derivatives fall in the class of affine models (see e.g. D'Amico, Kim and Wei \cite{D'AmicoKimWei2018}, 
Ho, Huang and Yildirim \cite{HoHuangYildirim2014} and 
Waldenberger \cite{Waldenberger2017}). Singor et al. 
\cite{Singor-etal2013} consider a Heston-type inflation 
model in combination with a Hull-White model for 
interest rates, with non-zero correlations. Hughston and 
Macrina \cite{HughstonMacrina2008} propose a discrete 
time model based on utility functions. Haubric et al. 
\cite{Haubric-etal2012} develop a discrete time model of 
nominal and real bond yield curves based on several 
stochastic drivers. 

Our model does not fall within the class of affine models, 
and does not reduce to other known models. 
Therefore a certain amount of mathematical work is 
needed to study it, in particular to derive the valuation 
equation for the price of a derivative and to show that it has a 
unique solution. In fact the proof relies on a general result 
by Costantini, Papi and D'Ippoliti 
\cite{CostantiniPapiD'Ippoliti2012} on valuation equations 
for jump-diffusion underlyings. The derivative 
payoff may depend on all three economic factors of the model and the discount factor is allowed to  
include inflation (see Remark \ref{discount}). 
Although in some special cases the price might admit a closed 
form, or might be approximated by a closed form, in 
general it has to be computed numerically. This can be 
done efficiently by the numerical 
algorithm described in Appendix C. 

In order to show that the higher complexity of the 
numerical implementation of our valuation model is 
justified, we compare our valuation model to the well known 
model of D'Amico, Kim and Wei \cite
{D'AmicoKimWei2018} and Ho, Huang and Yildirim 
\cite{HoHuangYildirim2014}). The improvement in the error 
measures is significant (see Tables \ref{tab:error} and 
\ref{tab:error2}).
Note that our model requires many fewer parameters than the Ho, 
Huang ad Yildirim model, which 
compensates the higher complexity of the numerical implementation, at 
least partly. 
Moreover, the fact that our model 
performs better than the Ho, Huang and Yildirim model 
with many fewer parameters suggests that it describes more 
closely the behaviour of the economic factors. 

The paper is organized as follows. In Section~\ref{sectionTheModel},
we introduce the mathematical model. 
In Section~\ref{sectionVE} we derive the valuation
equation. Proofs are postponed to Appendix B. The 
general result by Costantini, Papi and D'Ippoliti 
\cite{CostantiniPapiD'Ippoliti2012} is recalled in Appendix A. 
In Section \ref{sectionNS} we compare our model to the model of 
\cite{HoHuangYildirim2014}: We calibrate both models to 
the market prices of Zero Coupon Inflation Indexed Swaps from 
January 2008 to October 2015 and measure the performances of the two 
models in fitting the market data (Table \ref{tab:error}). We also carry out the 
same analysis for the period 2008-2011, when, 
due to the subprime crisis, interest rates dropped 
drastically and rapidly (Table \ref{tab:error2}).  
In both periods there is a significant improvement in the 
error measures.


\section{The model}\label{sectionTheModel}

\setcounter{equation}{0}

\begin{figure}[!htb]
\center{\includegraphics[scale=0.25]{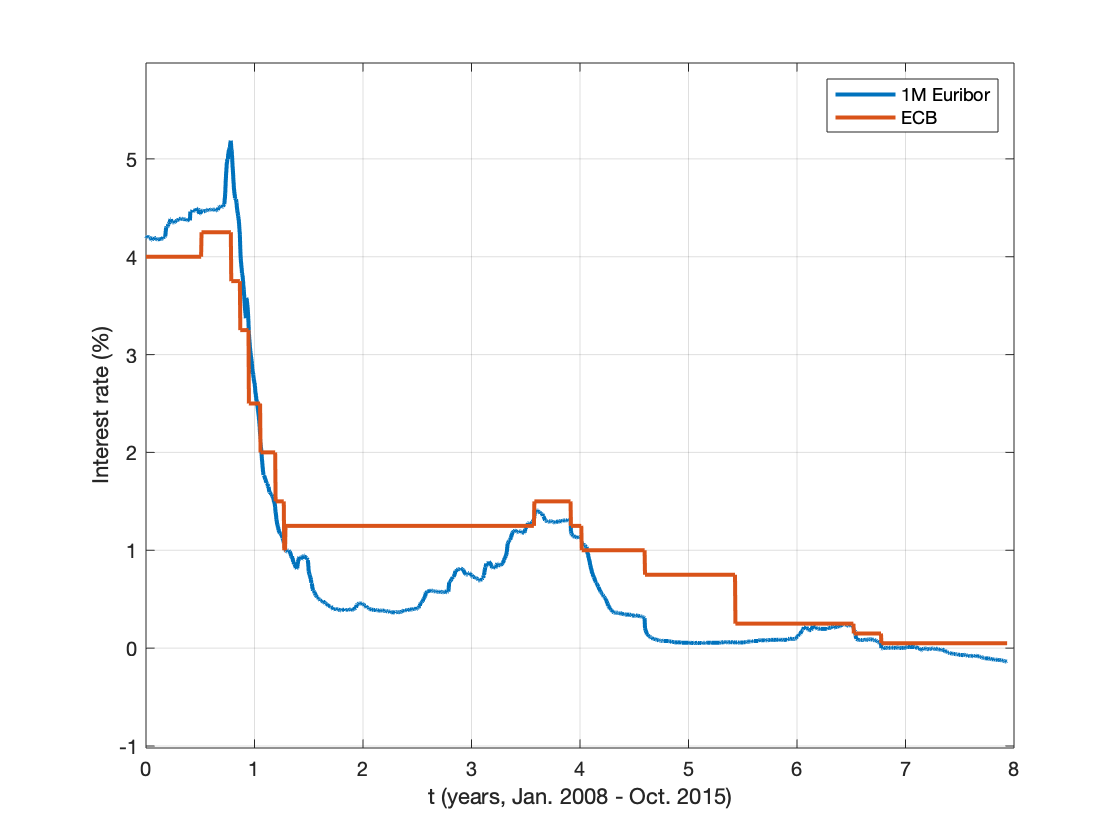}}
\caption[]{{\small Evolution of the official European Central Bank interest rate and the Euribor interest rate 
from January 2008 to October 2015.}}
\label{figrates}
\end{figure}
\begin{figure}[!htb]
\center{\includegraphics[scale=0.25]{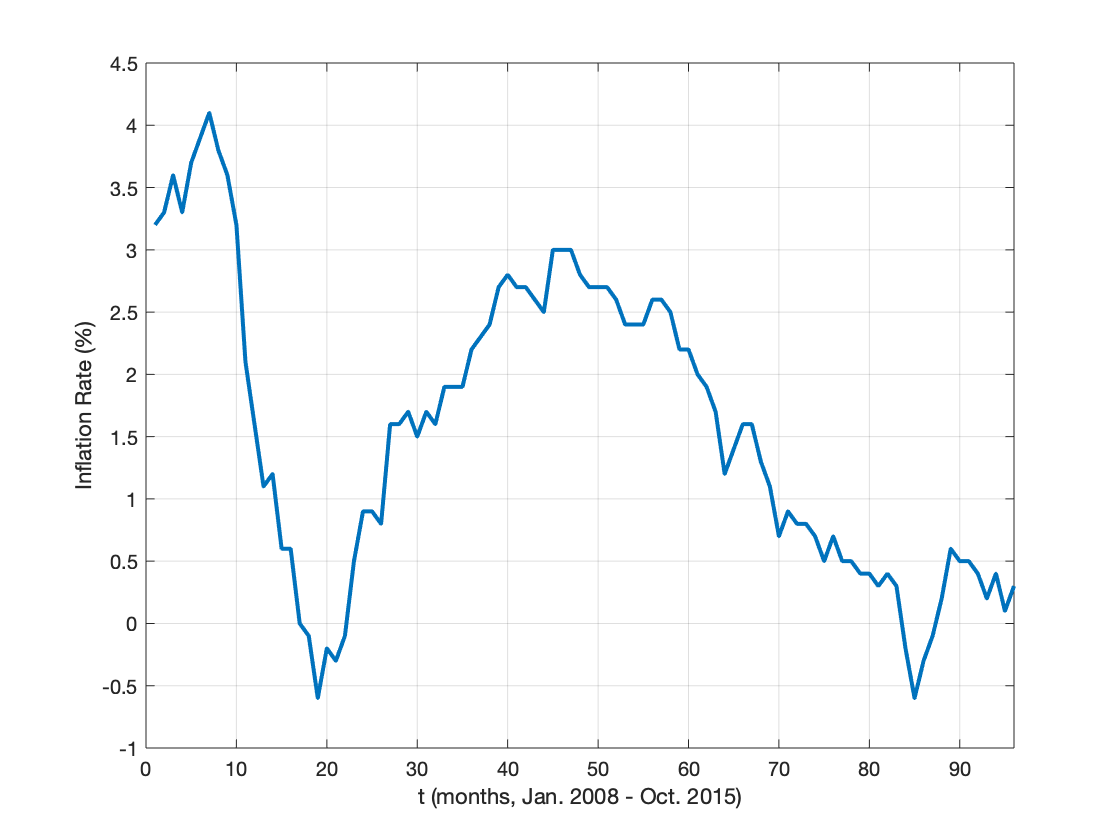}}
\caption[]{{\small Evolution of the European inflation rate from January 2008 to October 2015.}}
\label{figinfl}
\end{figure}

Figure~\ref{figrates} plots the ECB interest rate together 
with the Euribor interest rate between 2008 and 2015, while figure 
\ref{figinfl} plots inflation in the same period. 
The primary goal of the ECB is to maintain price stability,
i.e., to keep inflation within a desired range (close to 2\%). The
inflation target is achieved through periodic adjustments of the ECB
official interest rate and, consequently, of the short-term
interest rate. The goal of this paper, and specifically of 
this section, is to formulate a dynamical model  
that describes the interactions among inflation, 
the ECB and the short-term interest rates. 

From now on, we
fix the probability space $\left(\Omega ,\mathcal{F},\mathbb{P}\right
)$, 
where $\mathbb{P}$ is a martingale measure. 


\subsection{European Inflation}\label{subsecInflation}
European inflation data are officially made known
once a month, hence we model the inflation rate as a stochastic process that jumps
at fixed times, with jump sizes depending on its previous value 
and on the official ECB interest
rate, and is constant between two jumps. \\
Specifically, with the usual convention that one year is an
interval of length one, let $\mathcal{T}:=\{t_i\}_{i\geq 0,\ldots 
,M}$ 
be the sequence of times at which the values of
the inflation rate process $\left\{\Pi (t_i)\right\}_{t_i\in\mathcal{
T}}$ are
observed, where $t_0=0$, $t_1=\frac 1{12}$ and, for $i\geq 2$, $t_
i=it_1$.\\
The evolution is then given by
\begin{equation}\label{dynamicsInfl}\left\{\begin{array}{ll}
\Pi (0)=\Pi_0,\\
\Pi (t)=\Pi (t_i),&\quad t_i\leq t<t_{i+1},\\
\Pi (t_{i+1})=\gamma\left(\Pi (t_i),R(t_{i+1}^{-})\right)+\epsilon_{
i+1},&\quad t=t_{i+1},\end{array}
\right.\end{equation}
where $\gamma$ is a linear function defined by
\begin{equation}\gamma (\pi ,r)=\alpha\pi +k^{\Pi}(\pi^{
*}-\pi )+\beta\,r=(\alpha -k^{\Pi})\pi +k^{\Pi}\pi^{*}+\beta\,r,\label{gamma}\end{equation}
with $\alpha ,~\beta\in\mathbb{R}$ and $k^{\Pi},~\pi^{*}\in\mathbb{
R}_{+}$ constant parameters such that $0<\alpha -k^{\Pi}<1$.\\
The fluctuations $\{\epsilon_i\}_{i=1,\ldots ,M}$ are i.i.d. random variables
distributed according to the $\mathcal{N}(0,v^2)$ law, and
$R(t)$, for $t\in [0,T]$, is the
interest rate process which will be described in
Section~\ref{subsecECB}.\\
We can see that 
\begin{equation}\label{g}\gamma (\pi ,r)=\pi +\left(k^{\Pi}-\alpha 
+1\right)\left[\frac {k^{\Pi}\pi^{*}+\beta\,r}{k^{\Pi}-\alpha +1}
-\pi )\right]\end{equation}
and hence the condition $0<\alpha -k^{\Pi}<1$ yields that the process
$\Pi$  satisfies the mean-reversion property towards 
$\frac {k^{\Pi}\pi^{*}+\beta\,r}{k^{\Pi}-\alpha +1}$.


\subsection{European Central Bank Interest Rate}\label{subsecECB}

Looking at figure \ref{figrates},
we can recognise some important facts about the ECB interest rate.
The level of the rate is persistent, hence 
the sample path is a step function; The changes 
are multiples of 25 basis points (bp); A change is 
often followed by additional changes, frequently in the 
same direction. 
Therefore we model the ECB 
interest rate, $R(t)$, as a continuous time, 
pure jump process with finitely many possible upward and downward jump values. 
These jumps occur at random times $\{\vartheta_i\}$, and their size 
is equal to $k\delta$ with $\delta =0.0025$ and $k\in \{-m,...,-1
,1,...,m\}$.
The jump intensity, $\lambda$, is a function of the current values of 
inflation and the ECB interest rate (when the level of
the official interest rate is low, there is a tendency to avoid
further downward jumps). The probability of occurence of 
a jump $k\delta$ also depends upon the current values of inflation and the ECB 
interest rate, i.e. it is a function $p(\pi ,r,k\delta )$. 

Since, by definition, an interest rate is always larger than -1, we can assume,without loss 
of generality, $R(t)>\underline r,\quad\underline r\geq -1$ , for all $
t>0$.
In addition, we suppose that there exists a maximum value $0<\overline 
r<+\infty$ such that 
$R(t)<\overline r,$ for all $t>0.$
Consistently we assume that 
\begin{equation}\label{roverandunderline}p(\pi ,r,k\delta )=0\qquad\mbox{\rm for }
r+k\delta\notin (\underline r,\overline r),\end{equation}
and that 
\begin{equation}\label{overlineLambda}\overline {\lambda}:=\sup_{
(\pi ,r)\in\mathbb{R}\times (\underline r,\overline r)}\lambda (\pi 
,r)<+\infty ,\end{equation}
Of course we suppose that
\[\sum_{k\in \{-m,...,-1,1,...,m\}}p(\pi ,r,k\delta )=1\qquad\forall
\pi\in\mathbb{R},\,r\in [\underline r,\overline r]\]
Moreover in view of Section \ref{sectionVE}, we will make the following additional 
assumptions on 
$p(\cdot ,\cdot ,k\delta )$ and $\lambda =\lambda (\cdot ,\cdot )$: For $
k=-m,...,-1,1,...,m$, 
\begin{equation}\label{continuitapelambda}p(\cdot ,\cdot ,k\delta 
)\mbox{\rm \ and }\lambda =\lambda (\cdot ,\cdot )\mbox{\rm \ are continuous on }\mathbb{
R}\times [\underline r,\overline r]\end{equation}
and
\begin{equation}\label{derivatefinitepelambda}\begin{array}{c}
p(\pi ,\cdot ,k\delta )\mbox{\rm \ has a finite left derivative at }\overline 
r-k\delta\\
p(\pi ,\cdot ,-k\delta )\mbox{\rm \ has a finite right derivative at }\underline 
r+k\delta\end{array}
\end{equation}
An equivalent way of describing this jump process is to 
think that the process jumps with constant intensity $\overline {
\lambda}$, 
but the jump can be zero with probability 

\begin{equation}q(\pi ,r,0):=1-\frac {\lambda (\pi ,r)}{\overline {
\lambda}},\label{0jumpprob}\end{equation}
or can be $k\delta$ with probability 
\begin{equation}q(\pi ,r,k\delta ):=\frac {p(\pi ,r,k\delta )\lambda 
(\pi ,r)}{\overline {\lambda}},\qquad\quad k=-m,...,-1,1,...,m.\label{probabilities}\end{equation}
Then we can consider the ECB interest rate as the solution 
of the following stochastic equation 
\begin{equation}\label{dynamicsECB}R(t)=R_0+\int_0^tJ\left(\Pi (s^{
-}),R(s^{-}),U_{N(s^{-})+1}\right)~dN(s),\end{equation}
where $N=N(t)$ is a Poisson process  with intensity 
$\overline {\lambda}$, $\{U_n\}_{n\geq 0}$ are i.i.d. $[0,1]$-uniform random
variables, independent of $N$, and 
\begin{equation}\label{J}\begin{array}{ccc}
J(\pi ,r,u)&:=&-m\delta\mathbf{1}_{(0,1]}(q(\pi ,r,-m\delta ))\mathbf{
1}_{[0,q(\pi ,r,-m\delta )]}(u)\\
&&+\sum_{k=-m+1}^mk\delta~\mathbf{1}_{(0,1]}(q(\pi ,r,k\delta ))\mathbf{
1}_{(\sum_{h=-m}^{k-1}q(\pi ,r,h\delta ),\sum_{h=-m}^kq(\pi ,r,h\delta 
)]}(u),\end{array}
u\in [0,1],\end{equation}
the probabilities $q$ being defined by $(\ref{0jumpprob})$, 
$(\ref{probabilities})$. 


\subsection{Short-term Interest Rate}\label{subsecShort} 

As in many interest rate models in the literature, we 
suppose that  
the evolution of the short-term interest rate, $R^{sh}$, 
is a mean-reverting Ito process with coefficients 
depending on the ECB interest rate, $R$, and hence  
indirectly on the inflation $\Pi$ as well. 
More precisely, we suppose that $R^{sh}$ satisfies 
the following equation
\begin{equation}\label{dynamicsShort}\left\{\begin{array}{c}
dR^{sh}(t)=k^{sh}\left(b(R(t))-R^{sh}(t)\right)dt+\overline {\sigma}\left
(|R(t)-R^{sh}(t)|^2\right)\sqrt {|R^{sh}(t)|}dW(t),\\
R^{sh}(0)=R^{sh}_0,\end{array}
\right.\end{equation}
where $k^{sh}\in\mathbb{R}_{+}$ is a constant parameter and
$\{W_t\}_{t\in [0,T]}$ is a standard Wiener process. The function
$b(r)$ is defined by
\begin{equation}\label{b}b(r)=b_0+b_1r,\end{equation}
where $b_0,b_1\in\mathbb{R}$ are constant parameters and 
$\inf_{(\underline r,\overline r)}b(r)>0.$
The volatility coefficient $\overline {\sigma}$ is allowed to depend on the spread 
between $R^{sh}$ and $R$, so as to model, for instance, 
the fact that higher values of the spread 
may lead to higher volatility of $R^{sh}$.
In view 
of Section \ref{sectionVE}, the square of $\overline {\sigma}$, $\overline {
\sigma}^2$, satisfies the 
following assumptions:
\begin{equation}\label{sigmabarproperty1}\overline {\sigma}^2\in 
{\cal C}^2([0,\infty )),\end{equation}
\begin{equation}\label{sigmabarproperty2}0<\sigma_0^2\leq\overline {
\sigma}^2(q)\leq\sigma_1(1+\sqrt {q})\\
,\,\,\,q\in [0,+\infty )\\
.\end{equation}
For instance one can take 
\[\overline {\sigma}(q)=(1+q)^{1/4}.\]
We can ssume, without loss of generality, that $\overline {\sigma}(q)>0$ for all $q$.

Finally, in analogy to the CIR model, we assume that
\begin{equation}\label{bandsigmaproperty}k^{sh}\inf_{(\underline 
r,\overline r)}b(r)>\frac 12\overline {\sigma}^2\left(\overline 
r-(\underline r\wedge 0)\right).\end{equation}


\subsection{Well-posedness of the model}\label{subsecwellpos}
We suppose that the sources of randomness in (\ref{dynamicsInfl}),
(\ref{dynamicsECB}) and (\ref{dynamicsShort}), that is
$\{\epsilon_i\}_{1\leq i\leq M}$, $\{N(t)\}_{t\geq 0}$, $\{U_n\}_{
n\geq 1}$ and $\{W(t)\}_{t\geq 0}$, are mutually
independent. All information is given by the following filtration
\begin{equation}\label{filtrationbig}\mathcal{F}_t:=\sigma\left(\left
\{\Pi_0,R_0,R^{sh}_0,\epsilon_{I(s)},N(s),W(s),U_{N(s)},s\leq t\right
\}\right),\end{equation}
where $I(s)$ is the number of jumps of inflation up to time $s$,
namely,
\begin{equation}\label{I}I(s):=\max\{i\geq 0:t_i\leq s\},s\geq 0.\end{equation}
The model introduced in Sections
\ref{subsecInflation}-\ref{subsecShort} is well posed, in the sense
that there exists one and only one stochastic process
$\left(\Pi ,R,R^{sh}\right)$ verifying
(\ref{dynamicsInfl}), (\ref{dynamicsECB}) and (\ref{dynamicsShort}), 
as stated precisely in the following theorem, which is 
proven in Appendix A. 

\begin{teo}\label{theowellpos}
For every triple of $\mathbb{R}\times (\underline r,\overline r)\times 
(0,+\infty )$-valued r.v.'s $\left(\Pi_0,R_0,R^{sh}_0\right)$, 
there exists one and only one stochastic process
    $(\Pi ,R,R^{sh})$ defined on $(\Omega ,\mathcal{F},\mathbb{P}
)$, $\{\mathcal{F}_t\}$-adapted, 
such that $(\ref{dynamicsInfl})$, $(\ref{dynamicsECB})$ and 
$(\ref{dynamicsShort})$ 
are $\mathbb{P}$-a.s. verified. 
It holds $R^{sh}(t)>0$ for all $t\geq 0$, almost surely.
\end{teo}

\begin{proof}
See Appendix B.
\end{proof}


\section{The Valuation Equation}\label{sectionVE}

\setcounter{equation}{0}

In this section, we derive the valuation equation for 
a contingent claim with maturity $T$ and
payoff 
\begin{equation}Y(T)^p\,\Phi (\Pi (T),R(T),R^{sh}(T)),\quad p\geq 
0,\label{payoff}\end{equation}
where $Y$ is the inflation index, i.e. 
\begin{equation}Y(t)=\exp\bigg(\int_0^t\Pi (s)ds\bigg).\label{index}\end{equation}
It is well known that, under a risk neutral measure, 
the price $P(t)$ of such a contingent claim can
be expressed as the expected discounted payoff, namely, 
\[P(s)=\E\bigg[\exp\bigg(-\int_s^TR^{sh}(u)du\bigg)Y(T)^{\,p}\Phi 
(\Pi (T),R(T),R^{sh}(T))\bigg|{\cal F}_s\bigg],\]
where ${\cal F}_s$ is given by $(\ref{filtrationbig})$. 

\begin{rem}\label{discount}
Note that the 
form $(\ref{payoff})$-$(\ref{index})$ allows to consider a real discount 
factor, i.e. a discount factor that takes into account inflation. In fact 
\begin{eqnarray*}
&&\E\bigg[\exp\bigg(-\int_s^T\big(R^{sh}(u)-\Pi (u)\big)du\bigg)\Phi 
(\Pi (T),R(T),R^{sh}(T))\bigg|{\cal F}_s\bigg]\\
&=&\frac 1{Y(s)}\E\bigg[\exp\bigg(-\int_s^TR^{sh}(u)du\bigg)Y(T)\Phi 
(\Pi (T),R(T),R^{sh}(T))\bigg|{\cal F}_s\bigg].\end{eqnarray*}
\end{rem}
\vskip.3in

Due to the Markov property of $(Y,\Pi ,R,R^{sh})$, 
\[P(s)=\varphi (s,Y(s),\Pi (s),R(s),R^{sh}(s)),\]
where 
\begin{eqnarray}
&&\!\!\!\!\!\!\!\!\varphi (s,y,\pi ,r,z)=\nonumber\\
&&\!\!\!\!\!\!\!\!\!\!\!\!\E\bigg[\exp\bigg(-\int_s^TR^{sh}(u)du\bigg
)Y(T)^p\,\Phi (\Pi (T),R(T),R^{sh}(T))\bigg|(Y(s),\Pi (s),R(s),R^{
sh}(s))=(y,\pi ,r,z)\bigg].\qquad\label{price}\end{eqnarray}

We assume that $\Phi$ is continuous and satisfies 
\begin{equation}|\Phi (\pi ,r,z)|\leq C_0e^{C_1|\pi |}(1+z),\quad
\pi\in\R,\,r\in(\underline r,\overline r),\,\,\,z>0,\label{payoffgr}\end{equation}
for some constants $C_0,\,C_1\geq 0$. 
Supposing that $t_M\leq T<t_{M+1}$, we are going to show 
(Proposition \ref{recursion}) that 
\[\varphi (s,y,\pi ,r,z)=\left\{\begin{array}{ll}
y^p\,e^{p(T-s)\pi}\,\varphi^M(s-t_M,\pi ,r,z),&t_M\leq s\leq T,\\
y^p\,e^{p(t_{i+1}-s)\pi}\varphi^i(s-t_i,\pi ,r,z),&t_i\leq s<t_{i
+1},~i=0,\ldots ,M-1.\end{array}
\right.\]
where (Theorem \ref{sol}) for each inflation value, $\pi ,$ $\varphi^
i(\cdot ,\pi ,\cdot ,\cdot )$ is the unique 
solution of a terminal value problem for an equation that can be viewed as a simple 
parabolic Partial 
Integro-Differential Equation (with coefficients 
depending on the parameter $\pi$). The equation is the same 
for all $i$'s, but the terminal value  
is different for each $i$: it is $\Phi$ for $\varphi^M$ and it is defined recursively 
from $\varphi^{i+1}(0,\cdot ,\cdot ,\cdot )$, for $\varphi^i$, $i
=0,...,M-1$. 
As shown in Appendix C, this sequence of 
parametric terminal value problems can be solved 
numerically in an efficient way. 

To carry out our program, we introduce the 
parametrized process $(R^{\pi},R^{sh,\pi})$, obtained by ''freezing'' 
the inflation rate.  

\begin{pro}\label{param-wdef}

Let $(N,\{U_n\},W)$  be as in Theorem \ref{theowellpos}. 
For each $\pi\in\R$, for each $(\underline r,\overline r)\times (
0,\infty )$-valued random 
variable $(R^{\pi}_0,R^{sh,\pi}_0)$, independent of $(N,\{U_n\},W
)$, there exists 
a unique strong solution of the system of equations 
\begin{equation}\left\{\begin{array}{l}
R^{\pi}(t)=R^{\pi}_0+\int_0^tJ(\pi ,R^{\pi}(s^{-}),U_{N(s^{-})+1}
)dN(s),\\
R^{sh,\pi}(t)=R^{sh,\pi}_0+\int_0^tk^{sh}\left(b(R^{\pi}(s))-R^{s
h,\pi}(s)\right)ds\\
\qquad\qquad\qquad\quad +\int_0^t\overline {\sigma}\left(|R^{\pi}(s)-R^{
sh,\pi}(s)|^2\right)\sqrt {|R^{sh,\pi}(s)|}dW(s).\end{array}
\right.\label{param}\end{equation}

In addition, it holds $R^{\pi}(t)\in (\underline r,\overline r)$, $
R^{sh,\pi}(t)>0$ for all $t\geq 0.$ 

The solution corresponding to 
$(R^{\pi}_0,R^{sh,\pi}_0)=(r,z)$ will be denoted by $(R^{\pi}_r,R^{
sh,\pi}_{(r,z)})$. 

\end{pro}

\begin{proof}
The proof is analogous to the proof of Theorem 
\ref{theowellpos}.
\end{proof}

The following property of $(R^{\pi},R^{sh,\pi})$ will be used for 
the recursion. 

\begin{lemma}\label{param-prop}
For any continuous function 
$f:\R\times (\underline r,\overline r)\times (0,\infty )\rightarrow\R$  satisfying $
(\ref{payoffgr})$, \hfill\break 
for 
any $T>0$, 
$\E\bigg[e^{-\int_0^{T-t}R^{sh,\pi}_{(r,z)}(u)du}\big|f\big(\pi ,
R^{\pi}_r(T-t),R^{sh,\pi}_{(r,z)}(T-t)\big)\big|\bigg]$ is finite 
for every $(\pi ,r,z)$ and $0\leq t\leq T$, and the function 
\[F(t,\pi ,r,z)=\E\left[e^{-\int_0^{T-t}R^{sh,\pi}_{(r,z)}(u)du}f\left
(\pi ,R^{\pi}_r(T-t),R^{sh,\pi}_{(r,z)}(T-t)\right)\right]\]
is continuous on $[0,T]\times\R\times (\underline r,\overline r)\times 
(0,\infty )$ and satisfies 
$(\ref{payoffgr})$ uniformly for $t\in [0,T]$.
\end{lemma}

\begin{proof} See Appendix B. \end{proof}

For a continuous function $f:\R\times (\underline r,\overline r)\times 
(0,\infty )\rightarrow\R$  such that 
$|f(\pi ,r,z)|\leq C_0e^{C_1|\pi |}(1+z)$, let $\epsilon$ be a Gaussian variable with mean zero and 
variance $v^2$ and set 
\begin{equation}Bf(\pi ,r,z)=\E[f\left(\gamma (\pi ,r)+\epsilon ,
r,z\right)]=\frac 1{\sqrt {2\pi}v}\int_{\mathbb{
R}}f\left(\gamma (\pi ,r)+u,r,z\right)\exp\left(-\frac {u^2}{2v^
2}\right)du,\label{B}\end{equation}
where $\gamma$ is defined in $(\ref{gamma})$. 

Note that 
\[\E[e^{C_1|\gamma (\pi ,r)+\epsilon |}]\leq e^{C_1|\beta r+k^{\Pi}
\pi^{*}|}\,e^{C_1(\alpha -k^{\Pi})|\pi |}\E[e^{C_1|\epsilon |}]<\infty 
,\]
so that $Bf$ is continuous, by dominated convergence, and verifies the growth condition $(\ref{payoffgr})$, i.e. 
\begin{equation}|Bf(\pi ,r,z)|\leq C_0'e^{C_1'|\pi |}(1+z).\label{B-growth}\end{equation}

\begin{pro}\label{recursion}
Let $\varphi$ be the function defined by $(\ref{price})$, 
$t_M\leq T<t_{M+1}$. Then 
\begin{equation}\varphi (s,y,\pi ,r,z)=\left\{\begin{array}{lll}
y^{\,p}e^{p(T-s)\pi}\,\,\varphi^M(s-t_M,\pi ,r,z),&t_M\leq s\leq 
T,\\
y^pe^{p(t_{i+1}-s)\pi}\,\varphi^i(s-t_i,\pi ,r,z),&t_i\leq s<t_{i
+1},&i=0,...,M-1,\end{array}
\right.\label{phi}\end{equation}
where the functions $\varphi^i$ are defined recursively in the 
following way: 
\begin{equation}\varphi^M(t,\pi ,r,z)=\E\left[e^{-\int_0^{T-t_M-t}
R^{sh,\pi}_{(r,z)}(u)du}\Phi\left(\pi ,R^{\pi}_r(T-t_M-t),R^{sh,\pi}_{
(r,z)}(T-t_M-t)\right)\right],\quad 0\leq t\leq T-t_M,\label{phiM}\end{equation}
\begin{equation}\varphi^{M-1}(t,\pi ,r,z)=\E\left[e^{-\int_0^{t_1
-t}R^{sh,\pi}_{(r,z)}(u)du}B\left(e^{p(T-t_M)\cdot}\varphi^M(0,\cdot 
,\cdot ,\cdot )\right)\left(\pi ,R^{\pi}_r(t_1-t),R^{sh,\pi}_{(r,
z)}(t_1-t)\right)\right],\,\,\,0\leq t\leq t_1,\label{phiM-1}\end{equation}
and 
\begin{equation}\varphi^i(t,\pi ,r,z)=\E\left[e^{-\int_0^{t_1-t}R^{
sh,\pi}_{(r,z)}(u)du}B\left(e^{pt_1\cdot}\varphi^{i+1}(0,\cdot ,\cdot 
,\cdot )\right)\left(\pi ,R^{\pi}_r(t_1-t),R^{sh,\pi}_{(r,z)}(t_1
-t)\right)\right],\label{phii}\end{equation}
for $\,0\leq t\leq t_1\,$ and for $\,i=0,...,M-2$. 
\end{pro}

\begin{proof}
See Appendix B.
\end{proof}

For an $\R$-valued diffusion process $X$ with time 
independent coefficients $b$ 
and $\sigma$, we know, by the Feynman-Kac formula, 
under suitable assumptions, that the function 
\[\psi (t,x)=\E\bigg[\exp\bigg(-\int_0^{T-t}X_x(s)ds\bigg)\Psi (X_
x(T-t))\bigg],\]
where $X_x$ is the process starting at $x$, is of class ${\cal C}^{
1,2}$ and satisfies 
\[\frac {\partial\psi}{\partial t}(t,x)+b(x)\frac {\partial\psi}{
\partial x}(t,x)+\frac 12\sigma^2(x)\frac {\partial^2\psi}{\partial 
x^2}-x\psi (t,x)=0,\qquad 0\leq t<T,~x\in\R\\
,\]
\[\psi (T,x)=\Psi (x),\qquad x\in\R.\]
By analogy, we consider, for each fixed $\pi$, 
for each function $\varphi^i(\cdot ,\pi ,\cdot ,\cdot )$ 
defined in Proposition \ref{recursion}, the 
following equation, which reflects the dynamics of $(R^{\pi}_{(r,
z)},R^{sh,\pi}_{(r,z)})$: 
\begin{equation}\begin{array}{c}
\frac {\partial\psi}{\partial t}(t,r,z)+k^{sh}\left(b(r)-z\right)\frac {
\partial\psi}{\partial z}(t,r,z)+\frac 12\overline {\sigma}^2\left
(|r-z|^2\right)z\frac {\partial^2\psi}{\partial z^2}(t,r,z)\\
+\overline {\lambda}\sum_{k=-m}^m\left[\psi\left(t,r+k\delta ,z\right
)-\psi\left(t,r,z\right)\right]q(\pi ,r,k\delta )-z\psi (t,r,z)=0
,\end{array}
\label{valeq}\end{equation}
with terminal conditions 
\begin{eqnarray}
&\psi (T-t_M,\cdot,\cdot)=\Phi ,&\mbox{\rm for }\varphi^M,\non\\
&\psi (t_1,\cdot,\cdot)=B\left(e^{p(T-t_M)\cdot}\varphi^M(0,\cdot ,\cdot ,\cdot 
)\right),&\mbox{\rm for }\varphi^{M-1},\label{terminal}\\
&\psi (t_1,\cdot,\cdot)=B\left(e^{pt_1\cdot}\varphi^{i+1}(0,\cdot ,\cdot ,\cdot 
)\right),&\mbox{\rm for }\varphi^i,\,\,\,i=0,...,M-2.\non
\end{eqnarray}

\noindent$(\ref{valeq})$-$(\ref{terminal})$ is not a standard partial 
differential equation and it is not clear whether it admits a classical 
solution. However it turns out that 
$(\ref{valeq})$-$(\ref{terminal})$ admits 
a unique viscosity solution, which is enough to solve 
it numerically with the algorithm of Appendix C. 
As mentioned in the Introduction, 
we will obtain existence and uniqueness of the viscosity 
solution to $(\ref{valeq})$-$(\ref{terminal})$ 
from a general result for valuation 
equations for contingent claims written 
on jump-diffusion underlyings proved in 
\cite{CostantiniPapiD'Ippoliti2012} and summarised in 
Appendix A. 
Since the state space of $(\ref{valeq})$-$(\ref{terminal})$ is unbounded, as 
usual in the literature, uniqueness will hold in the class of functions with 
a prescribed growth rate.

\begin{teo}\label{sol}
Let $\varphi$ be the function defined by $(\ref{price})$, where the 
payoff is given by $(\ref{payoff})$-$(\ref{payoffgr})$, and let 
$\varphi^i$ be the functions defined in Proposition \ref{recursion}. 
Then: For $i=0,...,M$, for each $\pi\in\R$, 
the function $\varphi^i(\cdot ,\pi ,\cdot ,\cdot )$ 
is the unique viscosity solution of the equation 
$(\ref{valeq})$-$(\ref{terminal})$ satisfying the growth 
condition $(\ref{payoffgr})$, uniformly in time. 
\end{teo}

\begin{proof} The proof can be found in Appendix 
B.\end{proof}


\section{Calibration to ZCIIS and comparison with a 
benchmark model}\label{sectionNS} 

\setcounter{equation}{0}

We consider the affine-based, stochastic factor model of 
\cite{HoHuangYildirim2014} as our term of  comparison. 
Actually a setting such as theirs can be considered as a 
benchmark in empirical studies, as suggested 
by \cite{DuffiePanSingleton2000}. 

The comparison is focused on Inflation-Indexed Swaps which are swap contracts whose payoffs 
are linked to a specific inflation index. These instruments are mainly used for an inflation-risk-exposed 
investor to hedge against (or exchange) the inflation-risk undertaken. Among these derivatives, 
Zero-Coupon Inflation-Indexed Swaps are the most actively traded instruments and their quotations 
have been proven to provide additional information on inflation expectation. 

A Zero-Coupon Inflation-Indexed Swap (ZCIIS) contract is a bilateral
agreement that enables an investor or a hedger to secure an inflation-protected return with respect to an inflation index. 
The inflation receiver pays a predetermined fixed rate, and receives from the inflation seller inflation-linked payments. 
In the ZCIIS, starting at time $t_0$, with final time $T>t_0$, and nominal amount $
{\cal N}$, the fixed-leg payer pays ${\cal N}\left[(1+K_{{\rm Z}{\rm C}
{\rm I}{\rm I}{\rm S}}(t_0,T))^{T-t_0}-1\right]$ 
when the contract matures, $K_{{\rm Z}{\rm C}{\rm I}{\rm I}{\rm S}}
(t_0,T)$ being the contract 
fixed rate corresponding to the quotation in the market. The floating-leg payer pays 
${\cal N}\left[Y(T)/Y(t_0)-1\right]$, where $Y(t)$ is the value of the inflation 
index at time $t$.  

Following \cite{HoHuangYildirim2014} and also \cite{Mercurio2005}, by normalizing the nominal amount of the contract to be $
1$, 
under the usual no arbitrage condition, the fair swap rate 
for the ZCIIS starting at time $t_0$, with maturity $T$, is given by: 
\begin{equation}\label{ZCIISrate}K_{{\rm Z}{\rm C}{\rm I}{\rm I}{\rm S}}
(t_0,T)=\left[\frac {P_R(t_0,T)}{P_N(t_0,T)}\right]^{1/(T-t_0)}-1
,\end{equation}
where $P_N(t_0,T)$ denotes the price of the 
nominal zero-coupon, i.e. the zero-coupon, starting at $t_0$, that pays $
1$ at 
the maturity $T$, and $P_R(t_0,T)$ 
denotes the price of the 
real zero-coupon, i.e. the zero-coupon, starting at $t_0$, that pays 
$Y(T)/Y(t_0)$ at the maturity $T$. 

\subsection{Valuation of a ZCIIS in the benchmark model}\label{Ho.model}

In \cite{HoHuangYildirim2014}, the authors assume an 
affine-based model with an $n$-dimensional latent state  
vector $X(t)$ specified by a vector Vasicek process (see 
\cite{DuffiePanSingleton2000}):
\begin{equation}\label{Ho.model.1}dX(t)={\cal K}\,(\mu -X(t))dt+\Sigma\,
dW(t)\end{equation}
where ${\cal K}$, $\Sigma$ are $n\times n$ constant matrices, $\mu$ is a $
n\times 1$ vector and $W$ is a $n$-dimensional Wiener process. 
This implies that the state variable vector is mean-reverting with constant volatility. 
Following the specifications of \cite{D'AmicoKimWei2018}, 
the price of the real zero-coupon bond from $t$ to $T$ and its nominal 
counterpart can be expressed as: 
\begin{equation}\label{Ho.model.2}P_J(t,T)=\exp\left(A^J(T-t)+\langle 
B^J(T-t),X(t)\rangle\right),\quad J=R,N,\end{equation}
where $A^J(\tau )$, $B^J(\tau )$ solve a system of Riccati-type differential 
equations, whose coefficients involve the parameters ${\cal K}$, $
\mu$ 
and $\Sigma$ in $(\ref{Ho.model.1})$ and other ones. 
We refer readers to Section 2.2, page 161 in \cite{HoHuangYildirim2014} 
for a detailed description of their model. By 
$(\ref{Ho.model.2})$, $(\ref{ZCIISrate})$ takes the form 
\begin{equation}\label{Ho.model.3}K_{{\rm Z}{\rm C}{\rm I}{\rm I}
{\rm S}}(t_0,T)=\exp\left(\frac {A^R(T-t_0)-A^N(T-t_0)}{T-t_0}+\langle\frac {
B^R(T-t_0)-B^N(T-t_0)}{T-t_0},X(t_0)\rangle\right)-1.\end{equation}
Thus, for each 
set of values of the parameters of the model, the ZCIIS rate can be computed by solving 
numerically the systems of Riccati-type differential 
equations. 

\subsection{Valuation of a ZCIIS in our model}\label{ourmodel}

\noindent In order to implement the comparison, we deduce 
the ZCIIS rate under our modeling technique. 
We shall use the notation $(\cdot )^{{\rm m}{\rm o}{\rm d}{\rm e}
{\rm l}}$ to refer to our model and we assume that the nominal value of the contract is $
1$. 
In our model, the nominal zero-coupon bond price is given by
\begin{equation}\label{our.model.1}P_N^{{\rm m}{\rm o}{\rm d}{\rm e}
{\rm l}}(t_0,T)=\mathbb{E}\left[e^{-\int_{t_0}^TR^{sh}(s)ds}\,\Big
|{\cal F}_{t_0}\right].\end{equation}
and the real zero-coupon bond price is given by 
\begin{equation}\label{our.model.2}P_R^{{\rm m}{\rm o}{\rm d}{\rm e}
{\rm l}}(t_0,T)=\frac 1{Y(t_0)}\mathbb{E}\left[e^{-\int_{t_0}^TR^{
sh}(s)ds}Y(T)\Bigg|{\cal F}_{t_0}\right].\end{equation}

With the notation of Section \ref{sectionVE}, $P_N^{{\rm m}{\rm o}
{\rm d}{\rm e}{\rm l}}(t_0,T)$ is the 
price of a derivative with payoff of the form  
$(\ref{payoff})$-$(\ref{index})$ with $p=0$ and $\Phi (\pi ,r,z)=
1$, while 
$P_R^{{\rm m}{\rm o}{\rm d}{\rm e}{\rm l}}(t_0,T)$ is $\frac 1{Y(
t_0)}$ times the 
price of a derivative with payoff 
of the form  
$(\ref{payoff})$-$(\ref{index})$ with $p=1$ and $\Phi (\pi ,r,z)=
1$. 
Therefore, for each set of values of the parameters of 
the model,  both $P_N^{{\rm m}
{\rm o}{\rm d}{\rm e}{\rm l}}(t_0,T)$ and $P_R^{{\rm m}{\rm o}{\rm d}
{\rm e}{\rm l}}(t_0,T)$ 
can be computed by solving numerically the system of 
partial differential equations $(\ref{valeq})$-$(\ref{terminal})$. 
$K_{{\rm Z}{\rm C}{\rm I}{\rm I}{\rm S}}^{{\rm m}{\rm o}{\rm d}{\rm e}
{\rm l}}(t_0,T)$ can then be evaluated by $(\ref{ZCIISrate})$.

An efficient numerical scheme to solve 
the system of partial differential equations $(\ref{valeq})$-$(\ref{terminal})$ is described in Appendix C.

\subsection{Numerical tests}

We use financial market data for ZCIIS swap rates, which exist for the biggest 
European Monetary Union countries, and provide a valuable source of information. 
The dataset is provided by the Bloomberg platform and covers 
the period of time Jan 2008 to Oct 2015 
with a wide range of maturities (from 1 to 30 years).
Specifically, we consider a ZCIIS for the aggregate 
euro area which uses the Harmonized Index of Consumer Prices (HICP) as an indicator of inflation (excluding tobacco). 
The HICP index is available monthly and is obtained by Eurostat. For each month in the sample period, 
we consider the day where the HICP index is observed and we perform a cross-sectional estimation against 
the market swap rates for all maturities available in the dataset. 

We use two discrepancy measures documented in several works in the literature: the root mean-square error 
(RMSE) and the average relative prediction error (ARPE). They are defined as follows:
\begin{equation}ARPE(t_j)=\frac 1I\sum_{i=1}^I\frac {|K(t_j,t_j+T_i)-K^{{\rm m}
{\rm a}{\rm r}{\rm k}{\rm e}{\rm t}}_{ij}|}{K_{ij}^{{\rm m}{\rm a}
{\rm r}{\rm k}{\rm e}{\rm t}}},\end{equation}
\begin{equation}RMSE(t_j)=\sqrt {\sum_{i=1}^I\frac {|K(t_j,t_j+T_i)-K^{{\rm m}{\rm a}
{\rm r}{\rm k}{\rm e}{\rm t}}_{ij}|^2}I},\end{equation}
where $K(t_j,T_i)$ is the model implied swap rate at day $
t_j$, for the maturity $t_j+T_i$, for $i=1,\ldots I$, for $j=1,\ldots J$. 

For each model and for every $t_j$, we find an 
estimate for the model parameters by minimizing $RMSE^m(t_j)$. 

In order to compare the performances of the two 
models, for each of them we compute the above error measures using 
the estimated parameters and we take the average 
value over $\{t_j\}$:
\begin{eqnarray}
\overline {ARPE}&=&\frac 1J\sum_{j=1}^JARPE(t_j),\\
\overline {RMSE}&=&\frac 1J\sum_{j=1}^JRMSE(t_j).\end{eqnarray}
g
Concerning the parameters, for the 
\cite{HoHuangYildirim2014} model, following \cite{D'AmicoKimWei2018}
we use the following specification of the coefficients for 
the latent factor process $X(t)$:
\[\mu =\left[\begin{array}{c}
0\\
0\\
0\end{array}
\right],\qquad\Sigma =\left[\begin{array}{ccc}
0.01&0&0\\
\sigma_{21}&0.01&0\\
\sigma_{31}&\sigma_{32}&0.01\end{array}
\right],\qquad {\cal K}=\left[\begin{array}{ccc}
\kappa_{11}&0&0\\
0&\kappa_{22}&0\\
0&0&\kappa_{33}\end{array}
\right].\]
Hence $\sigma_{21}$, $\sigma_{31},\sigma_{32},$ $\kappa_{11}$, $\kappa_{
22},$ $\kappa_{33}$ are the parameters to be 
estimated. In addition, in the Riccati-type equations to be 
solved numerically, there are 11 more parameters. 
As usual in a cross-sectional estimation, 
the value of the state variable $X(t_j)$ is considered as a 
parameter in the model and is estimated jointly with 
all other parameters. Thus altogether there are 20 parameters to be 
calibrated. 

In our model, for the inflation index, in equation 
(\ref{gamma}), we set 
\[\alpha =1,\quad\pi^{\star}=\ln(1.02),\quad v=\sigma_{\Pi},\]
$ $where $\sigma_{\Pi}$ is the historical standard deviation of the 
monthly increments of the inflation rate 
determined by the HICP index, and we estimate the parameters 
\[\beta ,\quad k^{\Pi}.\]
For the ECB rate (Section \ref{subsecECB}), we choose 
\[\underline{r}=0.05\%,\,\,\,\overline{r}=4.5\%,\,\,\,m=1,\,\,\,\delta =0.25\%,\]
and the 
probabilities $q$ as 
\[q(\pi ,r,\delta )=\bigg[\bigg(\frac 1{0.3\sigma_{\Pi}}\big(\pi 
-(\pi^{*}+0.2\sigma_{\Pi})\big)\bigg)_{+}\wedge 1\bigg]\bigg[\bigg
(\frac 1{3\delta}\big((\overline r-\delta )-r\big)\bigg)_{+}\wedge 
1\bigg],\]
\[q(\pi ,r,-\delta )=\bigg[\bigg(\frac 1{0.3\sigma_{\Pi}}\big((\pi^{
*}-0.2\sigma_{\Pi})-\pi\big)\bigg)_{+}\wedge 1\bigg]\bigg[\bigg(\frac 
1{3\delta}\big(r-(\underline r+\delta )\big)\bigg)_{+}\wedge 1\bigg
],\]
\[q(\pi ,r,0)=1-q(\pi ,r,\delta )-q(\pi ,r,-\delta ),\]
and we estimate 
\[\overline {\lambda}.\]
Finally, for the short-term interest rate, we take the 
function $\overline {\sigma}$ constant, $\overline {\sigma}=\sigma_
0$, and we estimate 
all the parameters of equation $(\ref{dynamicsShort})$: 
\[k^{sh},\quad\sigma_0,\quad b_0,\quad b_1,\]
$ $with the constraints 
\[\label{our.model.sh}k^{sh}>0,\qquad\sigma_0>0,\qquad k^{sh}\big
(b_0+b_1\underline r\big)>\frac 12\sigma_0^2,\qquad k^{sh}\big(b_
0+b_1\overline r\big)>\frac 12\sigma_0^2.\]
In addition, the value of $R^{sh}(t_j)$ is considered as a 
parameter in the model and is estimated jointly with 
all other parameters. Thus altogether there are 8 parameters to be 
calibrated. 

We compare the performances of our model and of the 
model of \cite{HoHuangYildirim2014} in two periods: the 
whole period ranging from January 2008 to October 2015, 
and the period from January 2008 to November 2011, 
when, due to the subprime crisis, interest rates dropped 
from over $5\%$ to less than $0.5\%$ (see Figure \ref{figrates}). The results are 
summarized in Table 
\ref{tab:error} and Table \ref{tab:error2}, respectively. 

\begin{table}[h!]
\centering
\begin{tabular}{ l c c c c}
                    &$\overline {\textbf{RMSE}}$ &$\overline {\textbf{ARPE}}$\\
                        \hline
                             
\textbf{Our model} &0.1679 &0.01098   \\
\textbf{\cite{HoHuangYildirim2014} model}   &0.3146  &0.02101  \\
\hline
\end{tabular}
\caption {\em Error measures $\overline {ARPE}$  and $\overline {
RMSE}$  for the period Jan. 2008 - Oct. 2015 evaluated for our model and the 
\cite{ HoHuangYildirim2014}  model. The inflation swap rates in the sample are expressed in percentage.}
\label{tab:error}\end{table}

\begin{table}[h!]
\centering
\begin{tabular}{ l c c }
                    &$\overline {\textbf{RMSE}}$ &$\overline {\textbf{ARPE}}$\\
                        \hline
                             
\textbf{Our model} &0.19741  &0.004972 \\
\textbf{\cite{HoHuangYildirim2014} model}  &0.4271  &0.01076\\
\hline
\end{tabular}
\caption {\em Error measures $\overline {ARPE}$  and $\overline {RMSE}$  for the period Jan. 2008 - Nov. 2011 evaluated 
 for our model and the \cite{HoHuangYildirim2014} model. The inflation swap rates 
 in the sample are expressed in  percentage.}\label{tab:error2}
\end{table}

\section{Conclusions}
We have proposed a new model for the joint evolution of 
the inflation rate, the ECB interest rate and the short-term 
interest rate. We have derived a valuation equation 
that allows us to price inflation-linked derivatives by 
a numerical algorithm. We have compared our model to 
one of the best known models in the literature 
(\cite{D'AmicoKimWei2018} and \cite{HoHuangYildirim2014}): Calibrating both models to the same market data 
(ZCIIS from 2008 to 2015), the performance of our model in fitting the data 
appears to be significatively better than the \cite{HoHuangYildirim2014} 
model, although our model employs many fewer 
parameters.

\appendix
\appendixpage
\addappheadtotoc


\section{Viscosity solutions of integro-differential 
valuation equations}

\renewcommand {\theequation}{A.\arabic{equation}}
\setcounter{equation}{0}

For the convenience of the reader we summarise here a general result of \cite{CostantiniPapiD'Ippoliti2012}, on which the proof of Theorem \ref{sol} relies. 
\cite{CostantiniPapiD'Ippoliti2012} considers a general equation of the form
\begin{equation}\label{eq-1}\left\{\begin{array}{ll}
\partial_t\psi (t,x)+L\psi (t,x)-c(x)\psi (t,x)=g(t,x),\quad&(t,x
)\in (0,T)\times D,\\
\psi (T,x)=\Psi (x),\quad&x\in D,\end{array}
\right.\end{equation}
with
\begin{equation}\label{op}Lf(x)=\nabla f(x)b(x)+\frac 12\mbox{\rm tr}\left
(\nabla^2f(x)a(x)\right)+\int_D\left[f(x')-f(x)\right]m(x,dx'),\end{equation}
where $D$ is a (possibly unbounded) starshaped open subset of $\R^
d$.
The coefficients and data are assumed to satisfy the following 
conditions.

\begin{itemize}
\item[(H1)]\label{H1}
$a:D\rightarrow\R^{d\times d}$ is of the form $a=\sigma\sigma^T$, 
with $a=(a_{i,j})_{i,j=1,\ldots ,d}$, where $a_{i,j}\in {\cal C}^
2(D)$, 
and $b:D\rightarrow\R^d$ is Lipschitz continuous on compact subsets 
of $D$. 

\item[(H2)]\label{H2}
Denoting by ${\cal M}(D)$ the space of finite Borel measures on D, endowed with 
the weak convergence topology,
$m:D\rightarrow {\cal M}(D)$ is continuous and
\begin{eqnarray}
\sup_{x\in D}\left|\int_Df(x')m(x,dx')\right|<\infty ,\qquad\,\,\,
\forall\;f\in {\cal C}_c(D).\end{eqnarray}

\item[(H3)]\label{H3}
There exists a nonnegative function $V\in {\cal C}^2(D)$, such that
\begin{eqnarray}
&&\int_DV(x')m(x,dx')<+\infty ,\,\;\;\forall\;x\in D,\quad LV(x)\leq 
C\left(1+V(x)\right),\;\;\forall\;x\in D,\\
&&\lim_{x\in D,\,x\rightarrow x_0}V(x)=+\infty ,\,\forall\;x_0\in
\partial D,\quad\lim_{x\in D,\,|x|\rightarrow +\infty}V(x)=+\infty 
,\end{eqnarray}

\item[(H4)]\label{H4} $g\in {\cal C}([0,T]\times D)$, $c,\,\psi\in 
{\cal C}(D)$, and
$c$ is bounded from below. There exists a strictly increasing function $
l:[0,+\infty )\rightarrow [0,+\infty )$, such that
\begin{eqnarray}
&&s\mapsto sl(s)\;\quad\mbox{\rm is convex,}\qquad\lim_{s\rightarrow 
+\infty}l(s)=+\infty ,\\
&&(s_1+s_2)l(s_1+s_2)\leq C\left(s_1l(s_1)+s_2l(s_2)\right),\quad\;\;
\forall\;s_1,s_2\geq 0,\end{eqnarray}
and the following holds:
\begin{eqnarray}
|g(t,x)|l(|g(t,x)|)&&\leq C_T\left(1+V(x)\right),\\
|\Psi (x)|l(|\Psi (x)|)&&\leq C\left(1+V(x)\right),\end{eqnarray}
for all $(t,x)\in [0,T]\times D$.
\end{itemize}

We recall the definition of 
viscosity solution, in the present set up. 

\begin{defi}\label{viscsol}
A viscosity solution of -$(\ref{eq-1})$-$(\ref{op})$ is a continuous 
function $\psi$ defined on $[0,T]\times D$ such that 
\[\psi (T,x)=\Psi (x),\qquad x\in D,\]
and, for each $x\in D$, $t\in [0,T)$:

\item for every $f\in {\cal C}^{1,2}\big([0,T]\times D\big)$ such that 
\[\sup_{(\tilde {t},\tilde {x})\in [0,T]\times D}\left(\psi (\tilde 
t,\tilde x)-f(\tilde t,\tilde x)\right)=\left(\psi (t,x)-f(t,x)\right
)=0,\]
it holds 
\[\begin{array}{c}
\frac {\partial f}{\partial t}(t,x)+Lf(t,x)-c(x)\psi (t,x)\geq g(
t,x);\end{array}
\]
\item for every $f\in {\cal C}^{1,2}\big([0,T]\times D\big)$ such that 
\[\inf_{(\tilde {t},\tilde {x})\in [0,T]\times D}\left(\psi (\tilde 
t,\tilde x)-f(\tilde t,\tilde x)\right)=\left(\psi (t,x)-f(t,x)\right
)=0,\]
it holds 
\[\begin{array}{c}
\frac {\partial f}{\partial t}(t,x)+Lf(t,x)-c(x)\psi (t,x)\leq g(
t,x).\end{array}
\]
\end{defi}

The result proved in \cite{CostantiniPapiD'Ippoliti2012}  is the following. 

\begin{teo}(\cite{CostantiniPapiD'Ippoliti2012})\label{CPD}
Assume (H1), (H2), (H3) and (H4). 
Then for every $x\in D$, there exists
one and only one stochastic process, $X_x$, 
that solves the martingale problem for
$L$, $\mathcal D(L)=\mathcal C^2_c(D)$, with initial condition $x$. The function
\begin{eqnarray}
\psi (t,x)=\E\left[\Psi (X_x(T-t))e^{-\int_0^{T-t}c(X_x(r))dr}-\int_
0^{T-t}g(t+s,X_x(s))e^{-\int_0^sc(X_x(r))dr}ds\,\right],\label{fk}\\
(t,x)\in [0,T]\times D,\non\end{eqnarray}
\vskip.1in
\noindent
is the only viscosity solution of $(\ref{eq-1})$-$(\ref{op})$  satisfying
\begin{equation}|\psi (t,x)|l(|\psi (t,x)|)\leq C_T\left(1+V(x)\right
),\;\qquad\forall\;(t,x)\in [0,T]\times D.\label{V-growth}\end{equation}
\end{teo}\


\section{Proofs}

\renewcommand {\theequation}{B.\arabic{equation}}
\setcounter{equation}{0}

\noindent {\bf Proof of Theorem \ref{theowellpos}}
\vskip.1in

Setting $t_0:=0$ and $(\Pi (0),R(0),R^{sh}(0)):=\left(\Pi_0,R_0,R^{
sh}_0\right)$, we 
claim that, given a triple 
of $\mathbb{R}\times [\underline r,\overline r]\times\R$-valued, $
\{\mathcal{F}_{t_i}\}$-measurable r.v.'s 
$\left(\Pi (t_i),R(t_i),R^{sh}(t_i)\right)$, $(\Pi ,R,R^{sh})$ is pathwise uniquely 
defined on the interval $[t_i,t_{i+1}]$ and \break
$\left(\Pi (t_{i+1}),R(t_{i+1}),R^{sh}(t_{i+1})\right)$ is $\{\mathcal{
F}_{t_{i+1}}\}$-measurable. 
To see this, observe, first of all, that the probability that $N$ jumps at any of $
t_1,...,t_M$ is zero. 
Therefore, denoting by $\{\vartheta_n\}$ the jump times of 
$N$, $R$ can be defined simply in the following way: 
If $N(t_{i+1})>N(t_i)$, 
\begin{eqnarray*}
R(\vartheta_{N(t_i)+n}):=R(\vartheta_{N(t_i)+n-1}\vee t_i)+J\left
(\Pi (t_i),R(\vartheta_{N(t_i)+n-1}\vee t_i),U_{N(t_i)+n}\right),\\
\mbox{\rm \ for }1\leq n\leq N(t_{i+1})-N(t_i),\end{eqnarray*}
\[R(t):=R(\vartheta_{N(t_i)+n-1}\vee t_i)\quad\mbox{\rm \ for }\vartheta_{
N(t_i)+n-1}\vee t_i\leq t<\vartheta_{N(t_i)+n},\quad 1\leq n\leq 
N(t_{i+1})-N(t_i),\]
\[R(t):=R(\vartheta_{N(t_{i+1})})\qquad\mbox{\rm \ for }\vartheta_{
N(t_{i+1})}\leq t\leq t_{i+1}.\]
If $N(t_{i+1})=N(t_i)$,  
\[R(t):=R(t_i),\qquad\mbox{\rm \ for }t_i\leq t\leq t_{i+1}.\]
Note that $R(\vartheta_{N(t_i)+n}\wedge t_{i+1})$ is 
$\{\mathcal{F}_{\vartheta_{N(t_i)+n}\wedge t_{i+1}}\}$-measurable for all $
n\geq 1$, 
and that $R(t_{i+1})$ is $\{\mathcal{F}_{t_{i+1}}\}$-measurable.
Denote $\vartheta^i_0:=t_i$, $\vartheta^i_n:=\vartheta_{N(t_i)+n}
\wedge t_{i+1}$, $n\geq 1$. In each subinterval 
$[\vartheta^i_{n-1},\vartheta^i_n]$, we can write equation 
$(\ref{dynamicsShort})$ as 
\begin{eqnarray}
R^{sh}(\vartheta^i_{n-1}+t)=R^{sh}(\vartheta^i_{n-1})+\int_0^tk^{
sh}\left(b(R(\vartheta^i_{n-1}))-R^{sh}(\vartheta^i_{n-1}+s)\right
)ds\non\\
+\int_0^t\overline {\sigma}\left(|R(\vartheta^i_{n-1})-R^{sh}(\vartheta^
i_{n-1}+s)|^2\right)\sqrt {|R^{sh}(\vartheta^i_{n-1}+s)|}dW^{\vartheta^
i_{n-1}}(s)\label{IkedaWatanabe}\\
0\leq t\leq\vartheta^i_n-\vartheta^i_{n-1},\non\end{eqnarray}
where $W^{\vartheta^i_{n-1}}(s):=W(\vartheta^i_{n-1}+s)-W(\vartheta^
i_{n-1})$. 
Since $W$ is independent of $N$, $W^{\vartheta^i_{n-1}}$ is a standard 
Brownian motion, independent of $\vartheta^i_n-\vartheta^i_{n-1}$. Moreover, if 
$R^{sh}(\vartheta^i_{n-1})$ is $\{\mathcal{F}_{\vartheta^i_{n-1}}
\}$-measurable (hence
$\big(R(\vartheta^i_{n-1}),R^{sh}(\vartheta^i_{n-1})\big)$ is $\{\mathcal{
F}_{\vartheta^i_{n-1}}\}$-measurable) $W^{\vartheta^i_{n-1}}$ 
is independent of $\big(R(\vartheta^i_{n-1}),R^{sh}(\vartheta^i_{
n-1})\big)$. 
The diffusion coefficient in 
$(\ref{IkedaWatanabe})$ is locally Holder continuous by  
$(\ref{sigmabarproperty1})$, and has sublinear growth by 
$(\ref{sigmabarproperty2})$. Therefore, by the Corollary 
to Theorem 3.2, Chapter 4, of \cite{IkedaWatanabe}, there exists one and only one 
strong solution to $(\ref{IkedaWatanabe})$ 
(the Corollary to Theorem 3.2 of \cite{IkedaWatanabe} assumes global Holder 
continuity, but, as pointed out in the comment 
immediately preceding Theorem 3.2, its statement can be 
localized and it yields existence and uniqueness of the 
strong solution up to the explosion time; $(\ref{sigmabarproperty1})$ and 
Theorem 2.4, Chapter 4, of \cite{IkedaWatanabe} ensure that the 
explosion time is infinite). 
Then $R^{sh}(\vartheta^i_{n-1}+t)$ is pathwise uniquely defined for 
$0\leq t\leq\vartheta^i_n-\vartheta^i_{n-1}$ and $R^{sh}(\vartheta^
i_n)$ is $\{\mathcal{F}_{\vartheta^i_n}\}$-measurable. 
Since $R^{sh}(\vartheta^i_0)=R^{sh}(t_i)$ is $\{\mathcal{F}_{t_i}
\}$-measurable, i.e. 
$\{\mathcal{F}_{\vartheta^i_0}\}$-measurable, we see, by induction, that 
$R^{sh}$ is pathwise uniquely defined on $[t_i,t_{i+1}]$ and 
$R^{sh}(t_{i+1})$ is $\{\mathcal{F}_{t_{i+1}}\}$-measurable. 
By setting  
\[\Pi (t_{i+1})=\gamma\left(\Pi (t_i),R(t_{i+1}),R^{sh}(t_{i+1})\right
)+\epsilon_{i+1},\]
our claim is proved. 
Finally, let us show that $R^{sh}(t)>0$ for all $t\geq 0$, almost 
surely. Let $\alpha_n:=\inf\{t\geq 0:\,R^{sh}(t)\leq\frac 1n\}$. Then it will be 
enough to show, for every $z_0>0$, for $R^{sh}_0=z_0$, that
\[\P (\alpha_n\leq t)\rightarrow_{n\rightarrow\infty}0,\qquad\forall 
t>0.\]
To see this, consider, for $z>0$, the function 
\[V_1(z):=z^2-\ln(z).\]
We have 
\[\begin{array}{cc}
&k^{sh}\left(b(r)-z\right)V_1(z)'+\frac 12\overline {\sigma}^2\left
(|r-z|^2\right)V_1(z)^{\prime\prime}\\
&=k^{sh}\left(b(r)-z\right)\big(2z-\frac 1z\big)+\frac 12\overline {
\sigma}^2\left(|r-z|^2\right)\big(2z+\frac 1z\big)\\
&=\1_{\{z\leq\overline r\}}\frac 1z\bigg(\frac 12\overline {\sigma}^
2\left(|r-z|^2\right)-k^{sh}b(r)\bigg)+\1_{\{z>\overline r\}}\frac 
1z\bigg(\frac 12\overline {\sigma}^2\left(|r-z|^2\right)-k^{sh}b(
r)\bigg)\\
&\qquad\qquad\qquad\qquad +k^{sh}b(r)+2z\bigg(\frac 12\overline {
\sigma}^2\left(|r-z|^2\right)+k^{sh}b(r)-k^{sh}z\bigg)\\
&\leq\frac 1{2\overline r}\overline {\sigma}^2\left(|r-z|^2\right
)+k^{sh}b(r)+2z\bigg(\frac 12\overline {\sigma}^2\left(|r-z|^2\right
)+k^{sh}b(r)\bigg)\\
&\leq C(1+z^2),\end{array}
\]
where the last but one inequality follows from 
$(\ref{bandsigmaproperty})$ and $z>0$, 
and the last one follows from 
$(\ref{sigmabarproperty2})$. 
Let $\beta_k:=\inf\{t\geq 0:\,R^{sh}(t)\geq k\}$. By applying Ito's 
formula and taking expectations, we obtain 
\begin{eqnarray*}
&&\E[V_1(R^{sh}(t\wedge\alpha_n\wedge\beta_k)]\\
&&\leq V_1(z_0)+C\E\bigg[\int_0^{t\wedge\alpha_n\wedge\beta_k}\big
(1+R^{sh}(s)^2\big)ds\\
&&\leq V_1(z_0)+C\int_0^t\bigg(1+\E\bigg[V_1(R^{sh}(s\wedge\alpha_
n\wedge\beta_k)\bigg]\bigg)ds,\end{eqnarray*}
which implies, by Gronwall's Lemma and by taking limits 
as $k\rightarrow\infty$,  
\[\E[V_1(R^{sh}(t\wedge\alpha_n)]\leq\big(V_1(z_0)+Ct\big)e^{Ct},\]
and hence, 
\[\ln(n)\P (\alpha_n\leq t)\leq\big(V_1(z_0)+Ct\big)e^{Ct}.\]
\hfill$\Box$
\vskip.2in

\noindent {\bf Proof of Lemma \ref{param-prop}}
\vskip.1in

Let us show preliminarly that, for any $T>0$, 
\begin{equation}\E[R^{sh,\pi}_{(r,z)}(t)]\leq C_T(1+z),\qquad 0\leq 
t\leq T,\,\,\,\,\pi\in\R,\,\,\,r\in (\underline r,\overline r),\,\,\,
z>0,\label{mom1}\end{equation}
and, for every $q\geq 2$, 
\begin{equation}\E[R^{sh,\pi}_{(r,z)}(t)^q]\leq C_T(1+z^q),\qquad 
0\leq t\leq T,\,\,\,\,\pi\in\R,\,\,\,r\in (\underline r,\overline 
r),\,\,\,z>0.\label{momp}\end{equation}
In order to prove $(\ref{momp})$, consider a sequence of bounded, nonnegative $
{\cal C}^2$ 
functions $\{f_n\}$ such that $f_n(z)=z^q$ for 
$0<z\leq n$ and $f_n(z)\leq z^q$ for all $z>0$, and let $\alpha_n
:=\inf\{t\geq 0:\,R^{sh,\pi}_{(r,z)}(t)\geq n\}$. 
By applying Ito's Lemma to the semimartingale $R^{sh,\pi}_{(r,z)}$, 
and to the function $f_n$, and taking expectations, we 
obtain 
\begin{eqnarray*}
&&\E[R^{sh,\pi}_{(r,z)}(t\wedge\alpha_n)^q]\\
&=&\E[f_n(R^{sh,\pi}_{(r,z)}(t\wedge\alpha_n))]\\
&=&f_n(z)+k^{sh}\E\bigg[\int_0^{t\wedge\alpha_n}\bigg([(R^{\pi}_r
(s))-R^{sh,\pi}_{(r,z)}(s)]R^{sh,\pi}_{(r,z)}(s)^{q-1})\\
&&\qquad\qquad\qquad\qquad\qquad\qquad\qquad +\frac {q(q-1)}2\overline {
\sigma}^2(|R^{\pi}(s)-R^{sh,\pi}(s)|^2)R^{sh,\pi}(s)^{q-1}\bigg)d
s\bigg]\\
&&\leq z^q+C\E\bigg[\int_0^{t\wedge\alpha_n}\bigg(1+R^{sh,\pi}_{(
r,z)}(s))^q\bigg)ds\bigg]\\
&&\leq z^q+C\int_0^t\bigg(1+\E\bigg[R^{sh,\pi}_{(r,z)}(s\wedge\alpha_
n))^q\bigg]\bigg)ds,\end{eqnarray*}
where the last but one inequality follows from 
$(\ref{sigmabarproperty2})$. 

\noindent Therefore, by Gronwall's Lemma and Fatou's Lemma, 
\[\E[R^{sh,\pi}_{(r,z)}(t)^q]\leq\bigg(z^q+CT\bigg)e^{CT},\quad 0
\leq t\leq T.\]
$(\ref{mom1})$ can be proved in an analogous manner. 
\noindent$(\ref{mom1})$ yields both  
\[\E\bigg[e^{-\int_0^{T-t}R^{sh,\pi}_{(r,z)}(s)ds}\big|f\big(\pi 
,R^{\pi}_r(T-t),R^{sh,\pi}_{(r,z)}(T-t)\big)\big|\bigg]<\infty\;\;\mbox{\rm and}\;\;|F(t,\pi ,r,z)|\leq C_0'e^{C_1\pi}(1+
z),\]
for $(t,\pi ,r,z)\in [0,T]\times\R\times (\underline r,\overline 
r)\times (0,\infty )$. 
We are left with proving continuity of $F$. Let 
$(t_n,\pi_n,r_n,z_n)\rightarrow (t,\pi ,r,z)$. Then $\{R^{\pi}_{r_
n}\}$ converges to $R^{\pi}_r$ uniformly 
over compact time intervals, almost surely. In addition 
$\{R^{sh,\pi}_{(r_n,z_n)}\}$ is relatively compact by Theorems 3.8.6 and 
3.8.7 of \cite{EthierKurtz}, the Burkholder-Davies-Gundy 
inequality and $(\ref{momp})$ with $q=4$, and every limit point is 
continuous by Theorem 3.10.2 of \cite{EthierKurtz}. Therefore $\{
(R^{\pi}_{r_n},R^{sh,\pi}_{(r_n,z_n)})\}$ is 
relatively compact. By Theorem 2.7 of 
\cite{KurtzProtter}, every limit point of $\{(R^{\pi}_{r_n},R^{sh
,\pi}_{(r_n,z_n)})\}$ 
satisfies $(\ref{param})$ with $(R^{\pi}_0,R^{sh,\pi}_0)=(r,z)$. 
Since the solution to $(\ref{param})$ is (strongly and 
hence weakly) unique, we can conclude that 
$\{(R^{\pi}_{r_n},R^{sh,\pi}_{(r_n,z_n)})\}$ converges weakly to $
(R^{\pi}_r,R^{sh,\pi}_{(r,z)})$. 
The assertion then follows by observing that $(\ref{momp})$ 
implies that the random variables 
\[\bigg\{e^{-\int_0^{T-t_n}R^{sh,\pi}_{(r_n,z_n)}(s)ds}f\left(\pi_
n,R^{\pi}_{r_n}(T-t_n),R^{sh,\pi}_{(r_n,z_n)}(T-t_n)\right)\bigg\}\]
are uniformly integrable. 
\hfill$\Box$
\vskip.2in

\noindent {\bf Proof of Proposition \ref{recursion}}
\vskip.1in

Recall that $t_M\leq T<t_M+1$. For $t_M\leq s\leq T$, 
\begin{eqnarray*}
&&\E\bigg[\exp\bigg(-\int_s^TR^{sh}(u)du\bigg)Y(T)^p\Phi (\Pi (T)
,R(T),R^{sh}(T))\bigg|{\cal F}_s\bigg]\\
&=&\E\bigg[\exp\bigg(-\int_s^TR^{sh}(u)du\bigg)Y(s)^pe^{p(T-s)\Pi 
(s)}\,\Phi (\Pi (s),R(T),R^{sh}(T))\bigg|(Y(s),\Pi (s),R(s),R^{sh}
(s))\bigg]\end{eqnarray*}
Therefore, 
\begin{eqnarray*}
&&\varphi (s,y,\pi ,r,z)\\
&=&\E\bigg[\exp\bigg(-\int_s^TR^{sh}(u)du\bigg)Y(s)^pe^{p(T-s)\Pi 
(s)}\,\Phi (\Pi (s),R(T),R^{sh}(T))\bigg|(Y,\Pi ,R,R^{sh})(s)=(y,
\pi ,r,z)\bigg]\\
&=&y^pe^{p(T-s)\pi}\E\bigg[\exp\bigg(-\int_s^TR^{sh,\pi}(u)du\bigg
)\Phi (\pi ,R^{\pi}(T),R^{sh,\pi}(T))\bigg|(R^{\pi},R^{sh,\pi})(s
)=(r,z)\bigg]\\
&=&y^pe^{p(T-s)\pi}\E\bigg[\exp\bigg(-\int_0^{T-s}R^{sh,\pi}(u)du\bigg
)\Phi (\pi ,R^{\pi}(T-s),R^{sh,\pi}(T-s))\bigg|(R^{\pi},R^{sh,\pi}
)(0)=(r,z)\bigg]\\
&=&y^pe^{p(T-s)\pi}\varphi^M(s-t_M,\pi ,r,z),\end{eqnarray*}
where the last but one equality follows from the fact 
that $(R^{\pi},R^{sh,\pi})$ is time homogeneous. 

For $t_{M-1}\leq s<t_M$, 
\begin{eqnarray*}
&=&\E\bigg[\exp\bigg(-\int_s^TR^{sh}(u)du\bigg)Y(T)^p\Phi (\Pi (T
),R(T),R^{sh}(T))\bigg|{\cal F}_s\bigg]\\
&=&\E\bigg[\E\bigg[\exp\bigg(-\int_s^TR^{sh}(u)du\bigg)Y(T)^p\Phi 
(\Pi (T),R(T),R^{sh}(T))\bigg|{\cal F}_{t_M}\bigg]\bigg|{\cal F}_
s\bigg]\\
&=&\E\bigg[\exp\bigg(-\int_s^{t_M}R^{sh}(u)du\bigg)\E\bigg[\exp\bigg
(-\int_{t_M}^TR^{sh}(u)du\bigg)Y(T)^p\Phi (\Pi (T),R(T),R^{sh}(T)
)\bigg|{\cal F}_{t_M}\bigg]\bigg|{\cal F}_s\bigg]\\
&=&\E\bigg[\exp\bigg(-\int_s^{t_M}R^{sh}(u)du\bigg)\varphi (t_M,Y
(t_M),\Pi (t_M),R(t_M),R^{sh}(t_M))\bigg|{\cal F}_s\bigg]\\
&=&\E\bigg[\exp\bigg(-\int_s^{t_M}R^{sh}(u)du\bigg)Y(t_M)^pe^{p(T
-t_M)\Pi (t_M)}\varphi^M(0,\Pi (t_M),R(t_M),R^{sh}(t_M))\bigg|{\cal F}_
s\bigg]\\
&=&\E\bigg[\E\bigg[\exp\bigg(-\int_s^{t_M}R^{sh}(u)du\bigg)Y(t_M)^
pe^{p(T-t_M)\Pi (t_M)}\varphi^M(0,\Pi (t_M),R(t_M),R^{sh}(t_M))\bigg
|{\cal F}_{t_M^{-}}\bigg]\bigg|{\cal F}_s\bigg]\\
\ &=&\E\bigg[\exp\bigg(-\int_s^{t_M}R^{sh}(u)du\bigg)Y(t_M)^p\,B\left
(e^{p(T-t_M)\cdot}\varphi^M(0,\cdot ,\cdot ,\cdot )\right)(\Pi (s
),R(t_M),R^{sh}(t_M))\bigg|{\cal F}_s\bigg]\\
&=&\E\bigg[\exp\bigg(-\int_s^{t_M}R^{sh}(u)du\bigg)Y(s)^pe^{p(t_M
-s)\Pi (s)}\,B\left(e^{p(T-t_M)\cdot}\varphi^M(0,\cdot ,\cdot ,\cdot 
)\right)(\Pi (s),R(t_M),R^{sh}(t_M))\bigg|(Y,\Pi,R,R^{sh})(s)\bigg]\end{eqnarray*}
Therefore, 
\begin{eqnarray*}
&&\varphi (s,y,\pi ,r,z)\\
&=&\E\bigg[\exp\bigg(-\int_s^{t_M}R^{sh}(u)du\bigg)Y(s)^pe^{p(t_M
-s)\Pi (s)}\,B\left(e^{p(T-t_M)\cdot}\varphi^M(0,\cdot ,\cdot ,\cdot 
)\right)\bigg|(Y,\Pi ,R,R^{sh})(s)=(y,\pi ,r,z)\bigg]\\
&=&y^pe^{p(t_M-s)\pi}\E\bigg[\exp\bigg(-\int_s^{t_M}R^{sh,\pi}(u)
du\bigg)B\left(e^{p(T-t_M)\cdot}\varphi^M(0,\cdot ,\cdot ,\cdot )\right
)(\pi ,R^{\pi}(t_M),R^{sh,\pi}(t_M))\bigg|(R^{\pi},R^{sh,\pi})(t)
=(r,z)\bigg]\\
&=&y^pe^{p(t_M-s)\pi}\E\bigg[\exp\bigg(-\int_0^{t_M-s}R^{sh}_{(r,
z)}(u)du\bigg)B\left(e^{p(T-t_M)\cdot}\varphi^M(0,\cdot ,\cdot ,\cdot 
)\right)(\pi ,R^{\pi}_r(t_M-s),R^{sh,\pi}_{(r,z)}(t_M-s))\bigg]\\
&=&y^pe^{p(t_M-s)\pi}\varphi^{M-1}(s-t_{M-1},\pi ,r,z).\end{eqnarray*}

For $t_{i-1}\leq s<t_i$, $i\leq M-1$, assuming inductively that 
$\varphi (t_i,y,\pi ,r,z)=y^pe^{pt_1\pi}\varphi^i(0,\pi ,r,z)$, the same computations 
as for $i=M$ (with $t_M$ replaced by $t_i$ and $T$ replaced by 
$t_{i+1}$) yield 
\[\varphi (s,y,\pi ,r,z)=y^pe^{p(t_i-s)\pi}\varphi^{i-1}(s-t_{i-1}
,\pi ,r,z).\]
\hfill $\Box$
\vskip.2in

\noindent {\bf Proof of Theorem \ref{sol}}
\vskip.1in

In order to adjust to the general formulation of 
\cite{CostantiniPapiD'Ippoliti2012}, we note first of all that $(\ref{valeq})$ can 
be viewed as a simple Partial Integro-Differential 
Equation, namely 
\begin{equation}\begin{array}{c}
\frac {\partial\psi}{\partial t}(t,r,z)+k^{sh}\left(b(r)-z\right)\frac {
\partial\psi}{\partial z}(t,r,z)+\frac 12\overline {\sigma}^2\left
(|r-z|^2\right)z\frac {\partial^2\psi}{\partial z^2}(t,r,z)\\
\qquad\qquad\quad\quad +\int_{(\underline r,\overline r)}\left[\psi\left
(t,r',z\right)-\psi\left(t,r,z\right)\right]\mu (\pi ,r,dr')-z\psi 
(t,r,z)=0\end{array}
\label{PIDE}\end{equation}
where 
\begin{equation}\mu (\pi ,r,A):=\overline {\lambda}\sum_{k=-m}^m\1_
A(r+k\delta )q(\pi ,r,k\delta ),\qquad A\in {\cal B}\big((\underline 
r,\overline r)\big).\label{PIDE-mu}\end{equation}
Therefore, for each fixed $\pi$, $(\ref{PIDE})$ is of the form 
$\mbox{\rm (\ref{eq-1})-(\ref{op})}$ with 
\[\begin{array}{c}
L^{\pi}\psi (t,r,z)=k^{sh}\left(b(r)-z\right)\frac {\partial\psi}{
\partial z}(t,r,z)+\frac 12\overline {\sigma}^2\left(|r-z|^2\right
)z\frac {\partial^2\psi}{\partial z^2}(t,r,z)\\
\qquad\qquad\qquad\qquad\qquad +\int_{(\underline r,\overline r)}\left
[\psi\left(t,r',z\right)-\psi\left(t,r,z\right)\right]\mu (\pi ,r
,dr'),\end{array}
\]
\[g(r,z)=0,\qquad c(r,z)=z.\]

The proof thus consists in verifying the assumptions of 
\cite{CostantiniPapiD'Ippoliti2012}. 

Assumptions (H1), (H2) of theorem \ref{CPD} are satisfied by 
$(\ref{continuitapelambda})$ and 
$(\ref{sigmabarproperty1})$. 

As far as (H3) is concerned, it is sufficient to find, for 
each fixed $\pi$, 
$V_0^{\pi}\in {\cal C}^2(\underline r,\overline r)$, $V_1^{\pi}\in 
{\cal C}^2((0,\infty ))$ nonnegative and such that 
\begin{equation}\lim_{r\rightarrow\underline r^{+}}V_0^{\pi}(r)=\infty 
,\quad\lim_{r\rightarrow\overline r^{-}}V_0^{\pi}(r)=\infty ,\quad 
L^{\pi}V_0^{\pi}(r,z)\leq C^{_{_{\pi}}}(1+V_0^{\pi}(r)),\label{V0cond}\end{equation}
\begin{equation}\lim_{z\rightarrow 0^{+}}V_1^{\pi}(z)=\infty ,\quad\lim_{
z\rightarrow\infty}V_1^{\pi}(z)=\infty ,\quad L^{\pi}V_1^{\pi}(r,
z)\leq C^{\pi}(1+V_1^{\pi}(z)).\label{V1cond}\end{equation}
Then Assumption (H3) will be verified by 
\begin{equation}V^{\pi}(r,z)=V_0^{\pi}(r)+V_1^{\pi}(z).\label{V}\end{equation}
By the same computations as in the proof of Theorem 
\ref{theowellpos}, 
and by $(\ref{sigmabarproperty2})$,  
$(\ref{bandsigmaproperty})$, we can see that 
\begin{equation}V^{\pi}_1(z)=z^2-\ln(z).\label{V1}\end{equation}
satisfies $(\ref{V1cond})$.
For each fixed $\pi$, $V^{\pi}_0$ can be constructed in the following 
way. By $(\ref{derivatefinitepelambda})$, there exist 
$\underline {\beta}^{\pi}=\underline {\beta}\in {\cal C}^1([\underline 
r,\overline r])$, $\overline {\beta}^{\pi}=\overline {\beta}\in {\cal C}^
1([\underline r,\overline r])$ such that 
\[\underline {\beta}(\underline r)=0,\quad\underline {\beta}(r)>0\mbox{\rm \ for }
r>\underline r,\quad\underline {\beta}\mbox{\rm \ is nondecreasing}
,\]
\[\underline {\beta}(r)\geq\max_{h=1,..,m}p(\pi ,r+h\delta ,-h\delta 
),\]
and 
\[\overline {\beta}(\overline r)=0,\quad\overline {\beta}(r)>0\mbox{\rm \ for }
r<\,\overline r,\quad\overline {\beta}\mbox{\rm \ is nonincreasing}
,\]
\[\overline {\beta}(r)\geq\max_{h=1,..,m}p(\pi ,r-h\delta ,h\delta 
).\]
Setting 
\begin{equation}V^{\pi}_0(r):=\int_r^{\overline r}\frac 1{\underline {
\beta}(s)}\,ds\,+\,\int_{\underline r}^r\frac 1{\overline {\beta}
(s)}\,ds,\label{V0}\end{equation}
we have, for $k=1,...,m$, $r-k\delta\in (\underline r,\overline r
)$, 
\begin{eqnarray*}
&&\overline {\lambda}\left[V^{\pi}_0(r-k\delta )-V^{\pi}_0(r)\right
]q(\pi ,r,-k\delta )\\
&=&\lambda (\pi ,r)\left[\int_{r-k\delta}^r\frac 1{\underline {\beta}
(s)}\,ds\,-\,\int_{r-k\delta}^r\frac 1{\overline {\beta}(s)}\,ds\right
]p(\pi ,r,-k\delta )\\
&\leq&\overline {\lambda}\,\frac {m\delta}{\underline {\beta}(r-k
\delta )}p(\pi ,r,-k\delta )\,\\
&\leq&\overline {\lambda}m\delta .\end{eqnarray*}
Analogously, for $k=1,...,m$, $r+k\delta\in (\underline r,\overline 
r)$, 
\begin{eqnarray*}
&&\overline {\lambda}\left[V^{\pi}_0(r+k\delta )-V^{\pi}_0(r)\right
]q(\pi ,r,k\delta )\\
&=&\lambda (\pi ,r)\left[-\int_r^{r+k\delta}\frac 1{\underline {\beta}
(s)}\,ds\,+\,\int_r^{r+k\delta}\frac 1{\overline {\beta}(s)}\,ds\right
]p(\pi ,r,k\delta )\\
&\leq&\overline {\lambda}\,\frac {m\delta}{\overline {\beta}(r+k\delta 
)}p(\pi ,r,k\delta )\,\\
&\leq&\overline {\lambda}m\delta .\end{eqnarray*}
Thus $(\ref{V0cond})$ is satisfied. 
 
We now turn to assumption (H4). 
The conditions on $c(z):=z$ and $g(z)\equiv 0$ are clearly satisfied 
(note that the lower bound of $c$ needs not be positive). 
The terminal values are 
\begin{eqnarray}
\Psi^M&=&\Phi ,\non\\
\Psi^{M-1}&=&B\left(e^{p(T-t_M)\cdot}\varphi^M(0,\cdot ,\cdot ,\cdot 
)\right),\label{terminalPsi}\\
\Psi^i&=&B\left(e^{pt_1\cdot}\varphi^{i+1}(0,\cdot ,\cdot ,\cdot )\right
),\quad\mbox{\rm for }i=0,...,M-2.\non\end{eqnarray}
Taking the function $l$ as $l(s)=\sqrt {s}$,  
the condition on the terminal value $\Psi^i$ becomes 
\begin{equation}|\Psi^i(\pi ,r,z)|^{3/2}\leq C^{i,\pi}\left(1+V^{
\pi}(r,z)\right).\label{h4}\end{equation}
Suppose first that $t_M<T$. $\Psi^M=\Phi$ verifies the growth condition $
(\ref{payoffgr})$. 
Since 
\begin{equation}(1+z)^{3/2}\leq 4\big(1+z^2-\ln(z)\big)\leq4\left(1+V^{\pi}(r,z)\right),\label{ineq}\end{equation}
$\Psi^M$ satisfies $(\ref{h4})$ as well. Assuming inductively 
that $\Psi^{i+1}$
satisfies $(\ref{h4})$, by Theorem \ref{CPD}, 
$\varphi^{i+1}(\cdot ,\pi ,\cdot ,\cdot )$, defined in Proposition \ref{recursion}, 
is the unique viscosity solution of $(\ref{valeq})$-$(\ref{terminal})$
satisfying $(\ref{h4})$ uniformly in time. On the other hand, by 
Lemma \ref{param-prop}, $\varphi^{i+1}(\cdot ,\pi ,\cdot ,\cdot )$ verifies also 
$(\ref{payoffgr})$ uniformly in time, and hence, by 
$(\ref{ineq})$, is the unique 
viscosity solution of $(\ref{valeq})$-$(\ref{terminal})$ 
verifying $(\ref{payoffgr})$ uniformly in time. Moreover, 
by $(\ref{terminal})$ and $(\ref{B-growth})$, 
$\Psi^i$ satisfies 
$(\ref{payoffgr})$ and hence, by $(\ref{ineq})$, $(\ref{h4})$ as well. 

\noin If $t_M=T$, $\Psi^{M-1}=B\Phi$ satisfies $(\ref{h4})$ by 
$(\ref{B-growth})$ and $(\ref{ineq})$, and the induction 
starts form $M-1$. 
\hfill $\Box$

\section{Numerical scheme}\label{sectionTheImplementation}

\renewcommand {\theequation}{C.\arabic{equation}}
\setcounter{equation}{0}

In this section we present a finite difference scheme to solve numerically the pricing problem 
of an interest rate financial derivative under our model of 
Section \ref{sectionTheModel}. Let us remind that in 
recent years a great deal has been done for the numerical approximation of viscosity solutions 
for second order problems. In particular, we refer the reader to the fundamental paper by 
Barles and Souganidis \cite{BarlesSouganidis1991} who first showed convergence results for 
a large class of numerical schemes to the solution of fully nonlinear second order elliptic 
or parabolic partial differential equations. In addition we refer to \cite{BrianiNataliniChioma2004} 
for the extension of their arguments to the class of 
numerical schemes for integro-differential equations. 

Our numerical scheme is applied to the sequence of partial differential equations and their 
final conditions (\ref{valeq})-(\ref{terminal}) and is a modification of the scheme proposed in Zhu and Li 
\cite{ZhuLi2003}. An important point is that we do not 
impose any artificial condition at the 
boundary $z=0$: This is appropriate because of assumption 
(\ref{bandsigmaproperty}) and makes the scheme more 
accurate. 

For each interval $[t_i,t_{i+1}]$ and for every fixed value for 
the inflation rate $\pi$, we compute 
the solution of problem (\ref{valeq})-(\ref{terminal}) by the following method. We convert 
the problem into an initial-value problem letting $\tau =t_1-t$, 
and we compute the approximate value of the solution at $t_1-\tau_
n$, $\tau_n=n\Delta\tau$, $n=0,\ldots ,N$, $r_h=\underline r+h\delta$, $
h=1,\ldots ,H-1$, 
$z_j=j\Delta z$, $j=0,\ldots ,J$, where $\Delta\tau =t_1/N$, $\Delta 
z=z_{\max}/J$, $N$, $H$, $J$, being positive integers, 
such that $H\geq (\overline r-\underline r)/\delta >H-1$, and $z_{\max}
>>0$. Therefore, the numerical domain of the problem 
is $[0,t_1]\times [0,z_{\max}]$. Let $\psi^n_{h,j}$ denote the approximate value 
of the solution at the point $(t_1-\tau_n,r_h,z_j)$. 

With $r$ discretized as above, the non-local term in  
equation (\ref{valeq}) reduces to a linear operator in 
$\R^{H-1}$. 

As far as the partial differential equation in (\ref{valeq}) 
is concerned, at a point point $(\tau_n,r_h,z_j)$ with $z_j>0$, it can be discretized 
by the following second order approximation:
\begin{eqnarray}
\label{num.schema.int}&&\frac {\psi^{n+1}_{h,j}-\psi^n_{h,j}}{\Delta
\tau}=\frac {k^{sh}(b(r_h)-z_j)}{4\Delta z}\left(\psi^{n+1}_{h,j+
1}-\psi^{n+1}_{h,j-1}+\psi^n_{h,j+1}-\psi^n_{h,j-1}\right)+\non\\
&&\frac 1{4\Delta z^2}z_j\overline {\sigma}^2(|r_h-z_j|^2)\left(\psi^{
n+1}_{h,j+1}-2\psi^{n+1}_{h,j}+\psi^{n+1}_{h,j-1}+\psi^n_{h,j+1}-
2\psi^n_{h,j}+\psi^n_{h,j-1}\right)\\
&&+\overline {\lambda}\sum_{k=-\min(m,h-1)}^{\min(m,H-h-1)}\left[
\psi^n_{h+k,j}-\psi^n_{h,j}\right]q(\pi ,r_h,k\delta )-\frac {z_j}
2\left(\psi^{n+1}_{h,j}+\psi^n_{h,j}\right)\non,\end{eqnarray}
for any $h=1,\ldots ,H-1$, $j=1,\ldots J-1$, $n=0,\ldots ,N-1$. At the boundary $
z=0$, the partial differential equation 
in (\ref{valeq}) becomes a hyperbolic equation with respect to $z$, with a nonlocal term:
\begin{eqnarray}
\label{hyp.eq}\frac {\partial\psi}{\partial t}(t,r,z)=k^{sh}b(r)\frac {
\partial\psi}{\partial z}(t,r,z)+\overline {\lambda}\sum_{k=-m}^m\left
[\psi (t,r+k\delta ,z)-\psi (t,r,z)\right]q(\pi ,r,k\delta )\end{eqnarray}
Since $b(r)>0$, the value of $\psi$ on the boundary $z=0$ should be determined by the value of $
\psi$ 
inside the domain. Hence, we approximate 
$(\ref{hyp.eq})$ by the following scheme:
\begin{eqnarray}
\label{num.s.b}&&\frac {\psi^{n+1}_{h,0}-\psi^n_{h,0}}{\Delta\tau}
=\frac {k^{sh}b(r_h)}{4\Delta z}\left(-\psi^{n+1}_{h,2}+4\psi^{n+
1}_{h,1}-3\psi^{n+1}_{h,0}-\psi^n_{h,2}+4\psi^n_{h,1}-3\psi^n_{h,
0}\right)\nonumber\\
&&+\overline {\lambda}\sum_{k=-\min(m,h-1)}^{\min(m,H-h-1)}\left[
\psi^n_{h+k,0}-\psi^n_{h,0}\right]q(\pi ,r_h,k\delta ),\end{eqnarray}
for any $h=1,\ldots ,H-1$, $n=0,\ldots ,N-1$. Here $\partial\psi 
/\partial z$ is discretized by a one-side second order scheme 
so that all the node points involved are in the computational domain. Moreover we assign the initial
datum at $\psi^0_{h,j}=\psi (t_1 ,r_h,z_j)=\Psi^i(\pi ,r_h,z_j
)$, for any $j=0,\ldots ,J$. At the boundary $z=z_{\max}$ we adopt 
the Neumann boundary condition $\psi^n_{h,J}=\psi^n_{h,J-1}$, for any $
n=0,\ldots ,N$.
When $\psi^n_{h,j}$, $h=1,\ldots ,H-1$, $j=0,\ldots ,J$ are known from (\ref{num.schema.int}) and (\ref{num.s.b}), 
we can determine $\psi^{n+1}_{h,j}$, for any $h$ and $j$. Therefore, we can perform this procedure for $
n=0,\ldots ,N-1$ 
successively and finally find $\psi^N_{h,j}$, for any $h$ and $j$. Since truncation errors are second order everywhere, 
at least for a smooth enough solution it may be expected 
that the global error is $O(\Delta\tau ,\Delta z)$, 
see ~\cite{Marcozzi2001} and ~\cite{ZhuLi2003}.  
We can rewrite equations 
(\ref{num.schema.int}) and (\ref{num.s.b}) throughout using the following quantities:
\begin{eqnarray}
\label{coeff.nu}\nu_{h,j}=\frac {k^{sh}}4(b(r_h)-z_j)\frac {\Delta
\tau}{\Delta z},\qquad\qquad h=1,\ldots ,H-1,\;\;\;\ j=0,\ldots J
-1,\end{eqnarray}
\begin{eqnarray}
\label{coeff.xi}\xi_{h,j}&=&\frac {z_j}4\overline {\sigma}^2(|r_h
-z_j|^2)\frac {\Delta\tau}{(\Delta z)^2},\qquad\qquad h=1,\ldots 
,H-1,\;\;\;\ j=1,\ldots J-1,\\
\xi_{h,0}&=&\frac 34k^{sh}b(r_h)\frac {\Delta\tau}{\Delta z}.\end{eqnarray}
\begin{eqnarray}
\label{coeff.eta}\eta_{h,j}=\xi_{h,j}+\nu_{h,j},\qquad\;\;\theta_{
h,j}=\nu_{h,j}-\xi_{h,j},\qquad\;\;w_{h,j}=2\xi_{h,j}+\frac {\Delta
\tau}2z_j+1.\end{eqnarray}
In addition, for every $n=0,\ldots ,N-1$, $h=1,\ldots ,H-1$, $j=1
,\ldots ,J-1$, we define
\begin{eqnarray}
\label{coeff.Q}Q^n_{h,j}&=&\psi^n_{h,j}+\nu_{h,j}\left(\psi^n_{h,
j+1}-\psi^n_{h,j-1}\right)+\xi_{h,j}\left(\psi^n_{h,j+1}-2\psi^n_{
h,j}+\psi^n_{h,j-1}\right)\nonumber\\
&&+\Delta\tau\lambda (\pi ,r_h)\sum_{k=-\min(m,h)}^{\min(m,H-h)}\left
[\psi^n_{h+k,j}-\psi^n_{h,j}\right]p(\pi ,r_h,k\delta )-\frac {z_
j\Delta\tau}2\psi^n_{h,j},\\
Q^n_{h,0}&=&\psi^n_{h,0}+\nu_{h,0}\left(-\psi^n_{h,2}+4\psi^n_{h,
1}-3\psi^n_{h,0}\right)+\nonumber\\
&&+\Delta\tau\lambda (\pi ,r_h)\sum_{k=-\min(m,h)}^{\min(m,H-h)}\left
[\psi^n_{h+k,0}-\psi^n_{h,0}\right]p(\pi ,r_h,k\delta ).\end{eqnarray}
\begin{eqnarray}
\label{matrix}\!A_h\!=\!\left[\begin{array}{cccccc}
1+\xi_{h,0}&-4\nu_{h,0}&\nu_{h,0}&0&\cdots&0\\
\theta_{h,1}&w_{h,1}&-\eta_{h,1}&0&\cdots&0\\
0&\theta_{h,2}&w_{h,2}&-\eta_{h,2}&\cdots&0\\
\vdots&\vdots&\vdots&\vdots&\vdots&\vdots\\
0&0&0&0&\theta_{h,J-1}&(w_{h,J-1}-\eta_{h,J-1})\end{array}
\right]\;K_h^n=\left[\begin{array}{c}
Q^n_{h,0}\\
Q^n_{h,1}\\
\vdots\\
Q^n_{h,J-1}\\
\end{array}
\right],\end{eqnarray}
$A_h$ is a $J\times J$ matrix independent of $\psi^{n+1}_{h,.}$ and $
\psi^n_{h,.}$, whereas $K_h^n\in\mathbb{R}^J$ depends 
on the values of the numerical solution at the time step $n$. Therefore, keeping the terms 
which involve $\psi^{n+1}_{h,j}$, for $j=0,\ldots ,J-1$, on the left-hand side of equation 
$(\ref{num.schema.int})$, $(\ref{num.s.b})$  and bringing all the other terms on the right-hand side, 
we easily obtain the following linear system:
\begin{eqnarray}
\label{LS.1}A_h\psi^{n+1}_h=K_h^n\end{eqnarray}
for the computation of the numerical solution at the time step $n
+1$, given by
\begin{eqnarray}
\label{LS.2}\psi^{n+1}_h=\left[\begin{array}{c}
\psi^{n+1}_{h,0}\\
\psi^{n+1}_{h,1}\\
\vdots\\
\psi^{n+1}_{h,J-1}\\
\end{array}
\right],\end{eqnarray}
for any $h=1,\ldots ,H-1$. We observe that the coefficients in (\ref{coeff.eta}) satisfy
\begin{eqnarray}
\label{LS.3}w_{h,j}>|\theta_{h,j}|+|\eta_{h,j}|,\qquad\qquad\mbox{\rm for all $
j=1,\ldots J-1$},\end{eqnarray}
and the same holds for the coefficients in the first row of $A_h$. 
Therefore $A_h$ is strictly diagonally dominant, implying that $A_
h$ 
is invertible; moreover, since $w_{h,j}>1$, for any $j=1,\ldots ,
J-1$, 
the real parts of its eigenvalues are positive (these results follow 
from the well known Gershgorin's circle theorem). 
Therefore system (\ref{LS.1}) admits a unique solution.\\
For each discretized value $\pi$ of the observed inflation rate 
at time $t_i$, the numerical procedure allows to obtain 
$\psi^N_{h,j}=\psi^N_{h,j}(\pi )$, i.e. the approximate value of $
\varphi^i(0,\pi ,r_h,z_j)$, 
for each $h=1,\ldots ,H-1$, $j=0,\ldots ,J$, from the initial datum 
$\Psi^i$ evaluated at $(\pi ,r_{h'},z_{j'})$, $h'=1,\ldots ,H-1$, $
j'=0,\ldots ,J$. For 
$i=M-1,$ for each discretized value of $\pi$ and for each 
$h'=1,\ldots ,H-1$, $j'=0,\ldots ,J$, $\Psi^{M-1}(\pi ,r_{h'},z_{
j'})$ is obtained 
from the payoff $\Phi$ by applying a standard quadrature 
method for the evaluation of the integral operator $B$ 
defined in $(\ref{B})$. For $i<M-1$, $\Psi^i(\pi ,r_{h'},z_{j'})$ is 
obtained analogously from the approximate values of 
$e^{pt_1\pi'}\varphi^{i+1}(0,\pi',r_{h'},z_{j'})$, where $\pi'$ ranges over all discretized 
values of the inflation rate (the grid for the variable $u$ 
in the integral operator $B$ can be chosen so that 
$\gamma (\pi ,r_{h'},z_{j'})+u$ corresponds to a discretized value of the inflation 
rate).



\begin{thebibliography}{99}

\bibitem[BC97]{BaroneCastagna1997}
EMILIO BARONE and ANTONIO CASTAGNA.
\newblock The information content of tips.
\newblock {\em Internal report. SanPaolo IMI, Turin and Banca IMI, Milan},
  1997.

\bibitem[BRE81]{Bremaud}
PIERRE BREMAUD.
\newblock {\em Point Processes and Queues, Martingale Dynamics}.
\newblock Springer-Verlag, New York Heidelberg Berlin, 1981.

\bibitem[BT85]{BallTorous1985}
CLIFFORD~A. BALL and WALTER~N. TOROUS.
\newblock On jumps in common stock prices and their impact on call option
  pricing.
\newblock {\em The Journal of Finance}, 40(1):155--173, 1985.

\bibitem[CD02]{ChackoDas2002}
GEORGE CHACKO and SANJIV DAS.
\newblock Pricing interest rate derivatives: A general approach.
\newblock {\em The Review of Financial Studies}, 15(1):195--241, 2002.

\bibitem[CPD12]{CostantiniPapiD'Ippoliti2012}
CRISTINA COSTANTINI, MARCO PAPI, and FERNANDA D'IPPOLITI.
\newblock Singular risk-neutral valuation equations
\newblock {\em Finance and Stochastics}, 16(2):249--274, 2012.

\bibitem[CIR85]{CIRb}
JOHN~C. COX, JONATHAN~E. INGERSOLL, and STEPHEN~A. ROSS.
\newblock A theory of the term structure of interest rates.
\newblock {\em Econometrica}, 53(2):385--407, 1985.

\bibitem[DCPP97]{DeCecco1997}
MARCELLO DE~CECCO, LORENZO PECCHI, and GUSTAVO PIGA.
\newblock Managing public debt: index-linked bonds in theory and practice.
\newblock {\em Cheltenham: Edward Elgar}, 1997.

\bibitem[DD98]{DeaconDerry}
MARK DEACON and ANDREW DERRY.
\newblock {\em Inflation-Indexed Securities}.
\newblock Prentice Hall Europe, 1998.

\bibitem[EJ97]{EberleinJacod1997}
ERNST EBERLEIN and JEAN JACOD.
\newblock On the range of option prices.
\newblock {\em Finance and Stochastics}, 1(2):131--140, 1997.

\bibitem[EK86]{EthierKurtz}
STEWART~N. ETHIER and THOMAS~G. KURTZ.
\newblock {\em Markov Processes, Characterization and Convergence}.
\newblock John Wiley \& Sons, Inc., New York, 1986.

\bibitem[FMP00]{FaveroMissalePiga2000}
CARLO FAVERO, ALESSANDRO MISSALE, and GUSTAVO PIGA.
\newblock \textsc{EMU} and public debt management: One money, one debt?
\newblock {\em Center for Economic Policy Research, Policy Paper Series}, (3),
  2000.

\bibitem[FRI64]{Friedman}
AVNER FRIEDMAN.
\newblock {\em Partial Differential Equations of Parabolic Type}.
\newblock Prentice-Hall, Inc., 1964.

\bibitem[GM92]{GarroniMenaldi}
MARIA~GIOVANNA GARRONI and JOSE-LUIS MENALDI.
\newblock {\em Green functions for second order parabolic integro-differential
  problems}.
\newblock Research Notes in Mathematics 275, Longman Scientific \& Technical,
  Essex, 1992.

\bibitem[HP81]{HarrisonPliska1981}
MICHAEL~J. HARRISON and STANSLEY PLISKA.
\newblock Martingales and stochastic integrals in the theory of continuous
  trading.
\newblock {\em Stochastic Processes and Their Applications}, (11):215--260,
  1981.

\bibitem[HUG98]{Hughston1998}
LANE~P. HUGHSTON.
\newblock Inflation derivatives.
\newblock {\em Working Paper}, 1998.

\bibitem[IW81]{IkedaWatanabe}
NOBUYUKI IKEDA and SHINZO WATANABE.
\newblock {\em Stochastic Differential Equations and Diffusion Processes}.
\newblock North-Holland Publishing Company, Amsterdam Oxford New York, 1981.

\bibitem[JS87]{JacodShiryaev}
JEAN JACOD and ALBERT~N. SHIRYAEV.
\newblock {\em Limit theorems for Stochastic Processes}.
\newblock Springer-Verlag, Berlin Heidelberg New York London Paris Tokyo, 1987.

\bibitem[JY03]{JarrowYildirim2003}
ROBERT JARROW and YILDIRAY YILDIRIM.
\newblock Pricing treasury inflation protected securities and related
  derivatives using hjm model.
\newblock {\em Journal of Financial and Quantitative Analysis}, 38(2):409--430,
 2003.
 
 \bibitem[KP91]{KurtzProtter}
THOMAS~G. KURTZ and PHILIP PROTTER
\newblock  Weal limit theorems for stochastic integrals and stochastic 
differential equations
\newblock {\em Annals of probability}, 19(3):1035--1070, 1991
\newblock John Wiley \& Sons, Inc., New York, 1986.

\bibitem[PIA05]{Piazzesi2005}
MONIKA PIAZZESI.
\newblock Bond yields and the federal reserve.
\newblock {\em Journal of Political Economy}, 113(2):311--344, 2005.

\bibitem[PRO90]{Protter}
PHILIP PROTTER.
\newblock {\em Stochastic Integration and Differential Equations, A new
  Approach}.
\newblock Springer-Verlag, Berlin Heidelberg New York London Paris Tokyo Hong
  Kong, 1990.

\bibitem[RUN03]{Runggaldier2003}
WOLFGANG~J. RUNGGALDIER.
\newblock Jump diffusion models.
\newblock {\em Handbook of Heavy Tailed Distributions in Finance (S.T. Rachev,
  ed.)}, 1:169--209, 2003.

\bibitem[VAS77]{Vasicek1977}
OLDRICH~A. VASICEK.
\newblock An equilibrium characterization of the term structure.
\newblock {\em Journal of Financial Economics}, 5(2):177--188, 1977.

\end{thebibliography}


\vspace*{1cm}


\begin{thebibliography}{99}

\bibitem[BS91]{BarlesSouganidis1991}
BARLES G., SOUGANIDIS P.E.
\newblock{Convergence of approximation schemes for fully nonlinear equations,} 
\newblock {\em Asymptotic Analysis}, 4: 271-283, 1991. 

\bibitem[BLN04]{BrianiNataliniChioma2004}
BRIANI M., LA CHIOMA C., NATALINI R.
\newblock{Convergence of numerical schemes for viscosity solutions to integro-differential degenerate parabolic problems arising in financial theory,}
\newblock {\em Numerische Mathematik}, 98(4): 607-646, 2004.

\bibitem[CPD12]{CostantiniPapiD'Ippoliti2012}
COSTANTINI C., PAPI M., D'IPPOLITI F. 
\newblock{Singular risk-neutral valuation equations,}
\newblock {\em Finance and Stochastics}, 16(2):249--274, 2012.

\bibitem[DKW18]{D'AmicoKimWei2018}
D'AMICO S., KIM D.~H, WEI M.
\newblock{Tips from TIPS: The Informational Content of 
Treasury Inflation-Protected Security Prices,}
\newblock {\em Journal of Financial and Quantitative Analysis}, 53(1):395--436, 2018.

\bibitem[DDM04]{DeaconDerryMirfendereski2004}
DEACON M., DERRY A., MIRFENDERESKY D.
\newblock{Inflation indexed securities (2nd ed.)},
\newblock{Wiley Finance}, 2004.

\bibitem[DPS00]{DuffiePanSingleton2000}
DUFFIE D., PAN J., SINGLETON K.
\newblock{Transform analysis and asset prices for affine jump-diffusions,}
\newblock {\em Econometrica}, 68:1343--1376, 2000.

\bibitem[EK86]{EthierKurtz}
ETHIER S.~N., KURTZ T.~G.
\newblock{Markov Processes, Characterization and Convergence,}
\newblock{ John Wiley \& Sons}, Inc., New York, 1986.

\bibitem[HPR12]{Haubric-etal2012}
HAUBRIC J., PENNACCHI G., RITCHKEN P. 
\newblock{Expectations, Real Rates, and Risk Premia: Evidence from Inflation Swaps}, 
\newblock {\em The Review of Financial Studies}, 25(5): 1588--1629, 2012.  

\bibitem[HHY14]{HoHuangYildirim2014}
HO H.~W., HUANG H.~H.,YILDIRIM Y.
\newblock{Affine model of inflation-indexed derivatives and inflation risk premium}, 
\newblock {\em European Journal of Operational Research}, 235: 159-169, 2014. 

\bibitem[HM08]{HughstonMacrina2008}
HUGHSTON L.~P., MACRINA A. 
\newblock{Information, Inflation and Interest, Advances in Mathematics of Finance},
Banach Center Publications, 83, Institute of Mathematics Polish Academy of Sciences Warszawa, 2008.  

\bibitem[IW81]{IkedaWatanabe}
IKEDA N.,  WATANABE S.
\newblock {\em Stochastic Differential Equations and Diffusion Processes},
\newblock{North-Holland Publishing Company}, Amsterdam Oxford New York, 1981.

\bibitem[JY03]{JarrowYildirim2003}
JARROW R., YILDIRIM Y.
\newblock{Pricing treasury inflation protected securities and related derivatives using hjm model},
\newblock {\em Journal of Financial and Quantitative Analysis}, 38 (2): 409--430, 2003.
 
\bibitem[KP91]{KurtzProtter}
KURTZ T.~G., PROTTER P.
\newblock{Weal limit theorems for stochastic integrals and stochastic differential equations},
\newblock {\em Annals of probability}, 19 (3): 1035--1070, 1991.

\bibitem[MCC01]{Marcozzi2001}
MARCOZZI M.~D., CHOI S. and CHEN C.~S. 
\newblock{On the use of boundary conditions for variational 
formulations arising in financial mathematics},
\newblock {\em Appl. Math. Comp.}, 124: 197--214, 2001.

\bibitem[M05]{Mercurio2005}
MERCURIO F.
\newblock{Pricing inflation-indexed derivatives}, 
\newblock {\em Quantitative Finance}, 5 (3), 289--302, 2005. 

\bibitem[SGVBO13]{Singor-etal2013}
SINGOR S.~N. , GRZELAK L.~A., VAN BRAGT D.~D.~B., OOSTERLEE C.~W., 
\newblock{Pricing inflation products with stochastic volatility and stochastic interest rates}, \newblock {\em Insurance: Mathematics and Economics}, 52, 286-–299, 2013.   

\bibitem[W17]{Waldenberger2017}
WALDENBERGER S.
\newblock{The affine inflation market models},
\newblock {\em Applied Mathematical Finance}, 24 (4), 281--301, 2017. 

\bibitem[ZL03]{ZhuLi2003}
ZHU Y.~L., LI J.,
\newblock{Multi-factor financial derivatives on finite domains}, 
\newblock {\em Communications in Mathematical Sciences}, 1(2), 343--359, 2003.

\end{thebibliography}
\end{document}